\documentclass[a4paper,USenglish,cleveref, autoref, thm-restate]{lipics-v2021}
\usepackage{changepage}

\newenvironment{lp*}{\begin{equation*}  \begin{array}{lll}}{\end{array}\end{equation*}}

\providecommand{\tabularnewline}{\\}
\global\long\def\Otil{\tilde{O}}
\global\long\def\polylog{\operatorname{polylog}}
\global\long\def\poly{\operatorname{poly}}
\global\long\def\dist{\operatorname{dist}}

\newcommand{\E}{\mathbb{E}}

\renewcommand{\P}{\mathbb{P}}

\newcommand{\job}{\mathrm{job}}
\newcommand{\target}{\mathrm{target}}
\newcommand{\Sampling}{\textsc{Resample}}
\newcommand{\Resample}{\textsc{Resample}}

\newcommand{\Schedule}{\textsc{Schedule}}
\renewcommand{\exp}{\mathrm{exp}}

\def\ShowComment{True} % Switch comments  on or  off

\ifdefined\ShowComment
\def\thatchapholtext#1{\textcolor{purple}{#1}}
\def\thatchaphol#1{\marginpar{$\leftarrow$\fbox{T}}\footnote{$\Rightarrow$~{\sf\textcolor{purple}{#1 --Thatchaphol}}}}
\def\pattext#1{\textcolor{purple}{#1}}
\def\pat#1{\marginpar{$\leftarrow$\fbox{P}}\footnote{$\Rightarrow$~{\sf\textcolor{red}{#1 --Pat}}}}

\else
\def\thatchapholtext#1{}
\def\thatchaphol#1{}
\def\pattext#1{}
\def\pat#1{}

\fi

\pdfoutput=1 %uncomment to ensure pdflatex processing (mandatatory e.g. to submit to arXiv)
\hideLIPIcs  %uncomment to remove references to LIPIcs series (logo, DOI, ...), e.g. when preparing a pre-final version to be uploaded to arXiv or another public repository

\title{Simple Dynamic Spanners with Near-optimal Recourse against an Adaptive Adversary}

\author{Sayan Bhattacharya}{University of Warwick, United Kingdom}{bhattacharya@warrick.ac.uk}{}{}
\author{Thatchaphol Saranurak}{University of Michigan, Ann Arbor, USA}{thsa@umich.edu}{}{}
\author{Pattara Sukprasert}{Northwestern University, Evanston, USA}{pattara@u.northwestern.edu}{}{}

\authorrunning{S. Bhattacharya, T. Saranurak, and P. Sukprasert} %TODO mandatory. First: Use abbreviated first/middle names. Second (only in severe cases): Use first author plus 'et al.'

\Copyright{Sayan Bhattacharya, Thatchaphol Saranurak, and Pattara Sukprasert} %TODO mandatory, please use full first names. LIPIcs license is "CC-BY";  http://creativecommons.org/licenses/by/3.0/

 %TODO mandatory: Please choose ACM 2012 classifications from https://dl.acm.org/ccs/ccs_flat.cfm

 %TODO mandatory; please add comma-separated list of keywords

\funding{Sayan Bhattacharya is supported by Engineering and Physical Sciences Research Council, UK (EPSRC) Grant EP/S03353X/1.}%optional, to capture a funding statement, which applies to all authors. Please enter author specific funding statements as fifth argument of the \author macro.

\nolinenumbers %uncomment to disable line numbering

\begin{document}

\maketitle

\begin{abstract}
Designing dynamic algorithms against an adaptive adversary whose performance
match the ones assuming an oblivious adversary is a major research
program in the field of dynamic graph algorithms. One of the prominent
examples whose oblivious-vs-adaptive gap remains maximally large is
the \emph{fully dynamic spanner} problem; there exist algorithms assuming
an oblivious adversary with near-optimal size-stretch trade-off using
only $\polylog(n)$ update time {[}Baswana, Khurana, and Sarkar TALG'12;
Forster and Goranci STOC'19; Bernstein, Forster, and Henzinger SODA'20{]},
while against an adaptive adversary, even when we allow infinite time
and only count recourse (i.e.~the number of edge changes per update in the maintained spanner), all previous algorithms with stretch at most $\log^{5}(n)$
require at least $\Omega(n)$ amortized recourse {[}Ausiello, Franciosa,
and Italiano ESA'05{]}.

In this paper, we completely close this gap with respect to recourse
by showing algorithms against an adaptive adversary with near-optimal
size-stretch trade-off and recourse. More precisely, for any $k\ge1$,
our algorithm maintains a $(2k-1)$-spanner of size $O(n^{1+1/k}\log n)$
with $O(\log n)$ amortized recourse, which is optimal in all parameters
up to a $O(\log n)$ factor. As a step toward algorithms with small
update time (not just recourse), we show another algorithm that maintains
a $3$-spanner of size $\Otil(n^{1.5})$ with $\polylog(n)$ amortized
recourse \emph{and} simultaneously $\Otil(\sqrt{n})$ worst-case update
time. %
\begin{comment}
Our technique also implies as a by-product a distributed $3$-spanner
algorithm using 2 rounds in the CONGEST model, improving the best
algorithm with $\polylog(n)$ rounds {[}Rozho\v{n} and Ghaffari STOC'20{]}
\end{comment}
\end{abstract}

\section{Introduction}

Increasingly, algorithms are used interactively for data analysis,
decision making, and classically as data structures. Often it is not
realistic to assume that a user or an adversary is \emph{oblivious}
to the outputs of the algorithms; they can be \emph{adaptive} in the
sense that their updates and queries to the algorithm may depend on
the previous outputs they saw. Unfortunately, many classical algorithms
give strong guarantees only when assuming an oblivious adversary.
This calls for the design of algorithms that work against an adaptive
adversary whose performance match the ones assuming an oblivious adversary.
Driven by this question, there have been exciting lines of work across
different communities in theoretical computer science, including streaming
algorithms against an adaptive adversary \cite{ben2020framework,hasidim2020adversarially,woodruff2020tight,alon2021adversarial,kaplan2021separating,Braverman2021adversarial},
statistical algorithms against an adaptive data analyst \cite{hardt2014preventing,dwork2015preserving,bassily2021algorithmic,steinke2017tight},
and very recent algorithms for machine unlearning \cite{gupta2021adaptive}.

In the area of this paper, namely dynamic graph algorithms, a continuous
effort has also been put on designing algorithms against an adaptive
adversary. This is witnessed by dynamic algorithms for maintaining
spanning forests \cite{holm2001poly,NanongkaiS17,Wulff-Nilsen17,NanongkaiSW17,ChuzhoyGLNPS19},
shortest paths \cite{BernsteinC16,Bernstein17,BernsteinChechikSparse,ChuzhoyK19,ChuzhoyS20,gutenberg2020decremental,gutenberg2020deterministic,GutenbergWW20,Chuzhoy21},
matching \cite{BhattacharyaHI15,BhattacharyaHN16,BhattacharyaHN17,BhattacharyaK19,Wajc19,BhattacharyaK21deterministic},
and more. This development led to new powerful tools, such as the
expander decomposition and hierarchy \cite{SaranurakW19,GoranciRST20,liS2021}
applicable beyond dynamic algorithms \cite{Li21,li2021nearly,abboud2021apmf,zhang2021faster},
and other exciting applications such as the first almost-linear time
algorithms for many flow and cut problems \cite{BrandLNPSSSW20,BrandLLSS0W21,Chuzhoy21,BernsteinGS21}.
Nevertheless, for many fundamental dynamic graph problems, including
graph sparsifiers \cite{AbrahamDKKP16}, reachability \cite{BernsteinPW19},
directed shortest paths \cite{gutenberg2020decremental}, the performance
gap between algorithms against an oblivious and adaptive adversary
remains large, waiting to be explored and, hopefully, closed.

One of the most prominent dynamic problems whose oblivious-vs-adaptive
gap is maximally large is the \emph{fully dynamic spanner} problem
\cite{AusielloFI06,Elkin11,BaswanaKS12,BodwinK16,ForsterG19,BernsteinFH19,Bernstein2020fully}.
Given an unweighted undirected graph $G=(V,E)$ with $n$ vertices, an \emph{$\alpha$-spanner
}$H$ is a subgraph of $G$ whose pairwise distances between vertices
are preserved up to the \emph{stretch} factor of $\alpha$, i.e.,
for all $u,v\in V$, we have $\dist_{G}(u,v)\le\dist_{H}(u,v)\le\alpha\cdot\dist_{G}(u,v)$.\footnote{Here, $\dist_G(u,v)$ denotes the distance between $u$ and $v$ in graph $G$.}%
\begin{comment}
For any $k\ge1$, there exists a $(2k-1)$-spanner of size $O(n^{1+1/k})$
and this is tight assuming Erdos's conjecture.
\end{comment}
{} In this problem, we want to maintain an $\alpha$-spanner of a graph
$G$ while $G$ undergoes both edge insertions and deletions, and
for each edge update, spend as small update time as possible.

Assuming an oblivious adversary, near-optimal algorithms have been
shown: for every $k\ge1$, there are algorithms maintaining a $(2k-1)$-spanner
containing $\Otil(n^{1+1/k})$ edges\footnote{$\Otil(\cdot)$ hides a $\polylog(n)$ factor.},
which is nearly tight with the $\Omega(n^{1+1/k})$ bound from Erd\H{o}s' girth
conjecture (proven for the cases where $k=1,2,3,5$ \cite{wenger1991extremal}). Their update times are either $O(k\log^{2}n)$ amortized
\cite{BaswanaKS12,ForsterG19} or $O(1)^{k}\log^{3}n$ worst-case
\cite{BernsteinFH19}, both of which are polylogarithmic when $k$
is a constant.

In contrast, the only known dynamic spanner algorithm against an adaptive adversary
by \cite{AusielloFI06} requires $O(n)$ amortized update time and
it can maintain a $(2k-1)$-spanner of size $O(n^{1+1/k})$ only for
$k\le3$. Whether the $O(n)$ bound can be improved remained open until
very recently. Bernstein~et~al.~\cite{Bernstein2020fully} show
that a $\log^6(n)$-spanner can be maintained against an adaptive adversary using $\polylog(n)$
amortized update time. The drawback, however, is that their expander-based
technique is too crude to give any stretch smaller than $\polylog(n)$.
Hence, for $k\le \log^6(n)$, it is still unclear if the $\Theta(n)$
bound is inherent. Surprisingly, this holds even if we allow infinite
time, and only count \emph{recourse}, i.e., the number of edge changes per update
in the maintained spanner. The stark difference in performance between
the two adversarial settings motivates the main question of this paper:
\begin{center}
\emph{Is the $\Omega(n)$ recourse bound inherent for fully dynamic
spanners against an adaptive adversary? }
\par\end{center}

Recourse is an important performance measure of dynamic algorithms.
There are dynamic settings where changes in solutions are costly while
computation itself is considered cheap, and so the main goal is to
directly minimize recourse \cite{gupta2014maintaining,gupta2014online,avin2020dynamic,gupta2020fully}.
Even when the final goal is to minimize update time, there are many
dynamic algorithms that crucially require the design of subroutines
with recourse bounds \emph{stronger than} update time bounds to obtain
small final update time \cite{chechik2020dynamic,GoranciRST20,chen2020fast}.
Historically, there are dynamic problems, such as planar embedding
\cite{HolmR20soda,HolmR20stoc} and maximal independent set \cite{Censor-HillelHK16,BehnezhadDHSS19,ChechikZ19},
where low recourse algorithms were first discovered and later led
to fast update-time algorithms. Similar to dynamic spanners, there
are other fundamental problems, including topological sorting \cite{BernsteinC18cycle}
and edge coloring \cite{bhattacharya2021online}, for which low recourse
algorithms remain the crucial bottleneck to faster update time.

In this paper, we successfully break the $O(n)$ recourse barrier
and completely close the oblivious-vs-adaptive gap with respect to
recourse for fully dynamic spanners against an adaptive adversary.
\begin{theorem}
\label{thm:main greedy}There exists a deterministic algorithm that,
given an unweighted graph $G$ with $n$ vertices undergoing edge
insertions and deletions and a parameter $k\ge1$, maintains a $(2k-1)$-spanner
of $G$ containing $O(n^{1+1/k}\log n)$ edges using $O(\log n)$
amortized recourse.
\end{theorem}

As the above algorithm is deterministic, it automatically works against
an adaptive adversary. Each update can be processed in polynomial
time. Both the recourse and stretch-size trade-off of \Cref{thm:main greedy}
are optimal up to a $O(\log n)$ factor. When ignoring the update
time, it even dominates the current best algorithm assuming an oblivious adversary
\cite{BaswanaKS12,ForsterG19} that maintains a $(2k-1)$-spanner
of size $O(n^{1+1/k}\log n)$ using $O(k\log^{2}n)$ recourse. Therefore,
the oblivious-vs-adaptive gap for amortized recourse is closed.

The algorithm of \Cref{thm:main greedy} is as simple as possible.
As it turns out, a variant of the classical greedy spanner algorithm
\cite{AlthoferDDJS93} simply does the job! Although the argument
is short and ``obvious'' in hindsight, for us, it was very surprising.
This is because the greedy algorithm \emph{sequentially} inspects edges in
some fixed order, and its output solely depends on this order. Generally,
long chains of dependencies in algorithms are notoriously hard to analyze
in the dynamic setting. More recently, a similar greedy approach was
dynamized in the context of dynamic maximal independent set \cite{BehnezhadDHSS19}
by choosing a \emph{random order }for the greedy algorithm\emph{.
}Unfortunately, the random order must be kept secret from the adversary
and so this fails completely in our adaptive setting. Despite these
intuitive difficulties, our key insight is that we can \emph{adaptively
choose the order} for the greedy algorithm after each update. This
simple idea is enough, see \Cref{sec:greedy} for details.

\begin{comment}
because, algorithmically, Greedy algorithm seems very hard to dynamize.
\begin{thm}
$(2k-1)$-spanner optimal amortized recourse
\end{thm}

\begin{itemize}
\item Selling points
\begin{itemize}
\item First adaptive spanner with $o(n)$ recourse for $k<\polylog(n)$.
\item Everything is optimal (close the gap from \cite{BaswanaKS12,ForsterG19}).
In fact, the recourse is better!
\item As simple as possible.
\begin{itemize}
\item Greedy!
\item Although the argument is short, clean and so ``obvious'' in hindsight,
it was surprising for us because, algorithmically, Greedy algorithm
seems very hard to dynamize.
\begin{itemize}
\item Example: maximal matching, MIS (by soheil): Greedy is so sequential,
Randomized Greedy is dynamic, but inherently oblivious.
\item (Which is why no previous dynamic is nowhere close to Greedy).
\end{itemize}
\item It turns out that when the computation is ignored, Greedy spanners
are very stable.
\end{itemize}
\end{itemize}
\end{itemize}
\end{comment}

\Cref{thm:main greedy} leaves open the oblivious-vs-adaptive gap for
the update time. Below, we show a partial progress on this direction
by showing an algorithm with near-optimal recourse and simultaneously
non-trivial update time.
\begin{theorem}
\label{thm:main proactive}There exists a randomized algorithm that,
given an unweighted graph $G$ with $n$ vertices undergoing edge
insertions and deletions, with high probability maintains against
an adaptive adversary a $3$-spanner of $G$ containing $\Otil(n^{1.5})$
edges using $\Otil(1)$ amortized recourse \emph{and} $\Otil(\sqrt{n})$
worst-case update time.
\end{theorem}

We note again that, prior to the above result, there was no algorithm
against an adaptive adversary with $o(n)$ \emph{amortized update
time} that can maintain a spanner of stretch less than $\log^6(n)$.
\Cref{thm:main proactive} shows that for $3$-spanners, the update
time can be $\Otil(\sqrt{n})$ worst-case, while guaranteeing near-optimal
recourse. %
\begin{comment}
In the light of \ref{thm:main proactive}, we are quite optimistic
that there exists a near-optimal algorithms with polylogarithmic update
time against an adaptive adversary.
\end{comment}
{}

%The basic idea of \Cref{thm:main proactive} resembles the static spanner construction by Grossman and Parter \cite{GrossmanP17} for distributed computation.
%The key technique for proving \Cref{thm:main proactive} is called \emph{proactive resampling}, recently introduced in \cite{Bernstein2020fully} as a new technique for handling an adaptive adversary. We apply it on a modification of a spanner construction by Grossman and Parter  \cite{GrossmanP17} from distributed computation community. The modification is small but seems necessary.
We prove \Cref{thm:main proactive} by employing a technique called {\em proactive resampling}, which was recently introduced in \cite{Bernstein2020fully} for handling an adaptive adversary. We apply this technique on a modification of a spanner construction by Grossman and Parter  \cite{GrossmanP17} from distributed computation community. The modification is small, but seems inherently necessary for bounding the recourse.
%
%
%
%Proactive resampling is a new technique for handling an adaptive adversary. It is conceptually simple, quite generic, and inherently randomized. This is in contrast to known generic techniques (e.g.~expander-based techniques) that work against an adaptive adversary because they can be made deterministic at the end.

To successfully apply proactive resampling, we refine the technique in two ways.
First, we present a simple abstraction in terms of a certain load balancing problem that captures the power of proactive resampling. Previously, the technique was presented and applied specifically for the dynamic cut sparsifier problem \cite{Bernstein2020fully}. But actually, this technique is conceptually simple and quite generic, so our new abstraction will likely facilitate future applications.
%\footnote{Strictly speaking, as we aim for simplicity, our abstraction does not fully capture the dynamic cut sparsifier application in \cite{Bernstein2020fully} yet. However, we believe this can be done with a rather straightforward generalization.
%}
%
Our second technical contribution is to generalize and make the proactive resampling technique more flexible. At a very high level, in \cite{Bernstein2020fully}, there is a single parameter about sampling probability that is \emph{fixed} throughout the whole process, and their analysis requires this fact. In our load-balancing abstraction, we need to work with multiple sampling probabilities and, moreover, they change through time. We manage to analyze this generalized process using probabilistic tools about \emph{stochastic domination}, which in turn  simplifies the whole analysis.

%Our analysis of proactive resampling is simpler and more general than what appeared in ....
%The first issue we face is that we sample each edge with different probability.
%Even worse, this probability is a function of a node's degree at a certain timestep.
%As we deal with decremental graph, this function is affected by the adversary's actions.
%Our goal is to analyze that the edges we see adjacent to a node at any timestep
%should be concentrated on the expectation, even if we don't know what edges would appear at that timestep.
%To us, the sampling probability, the number of degrees, the expectations, and the \emph{relevant} edges are all not known prior to
%that timestep. How can we argue that such a thing is concentrated?
%This is how the notion of \emph{stochastic domination} comes into play.
%We eventually define an independent sequence of random variables that \emph{dominates} the relevant sequence that we might see
%and relate this sequence to the expected degree of nodes. This works even though the length of the sequence is also a random variable.
%While the analysis has many steps, each step is quite straightforward.

If a strong recourse bound is not needed,
then proactive resampling can be bypassed and the algorithm becomes very simple, deterministic, and has slightly
improved bounds as follows.
\begin{theorem}
\label{thm:main det worst case}There exists a deterministic algorithm
that, given an unweighted graph $G$ with $n$ vertices undergoing
edge insertions and deletions, maintains a $3$-spanner of $G$ containing
$O(n^{1.5})$ edges using $O(\min\{\Delta,\sqrt{n}\}\log n)$ worst-case
update time, where $\Delta$ is the maximum degree of $G$.
\end{theorem}

Despite its simplicity, the above result improves the update time
of the fastest deterministic dynamic 3-spanner algorithm \cite{AusielloFI06}
from $O(\Delta)$ amortized to $O(\min\{\Delta,\sqrt{n}\}\log n)$
worst-case. In fact, all previous dynamic spanner algorithms with
worst-case update time either assume an oblivious adversary \cite{Elkin11,BodwinK16,BernsteinFH19}
or have a very large stretch of $n^{o(1)}$ \cite{Bernstein2020fully}.
See \Cref{tab:state of the art} for detailed comparison.

\begin{table}
    \begin{adjustwidth}{-.5cm}{}
\footnotesize{
\begin{centering}
\begin{tabular}{|l|c|c|c|c|c|}
\hline
\textbf{Ref.}  & \textbf{Stretch}  & \textbf{Size}  & \textbf{Recourse} & \textbf{Update Time} & \textbf{Deterministic?}\tabularnewline
\hline
\hline
\multicolumn{6}{|l|}{\textbf{Against an oblivious adversary}}\tabularnewline
\hline
\cite{BaswanaKS12} & $2k-1$  & $O(k^{8}n^{1+1/k}\log^{2}n)$  & \multicolumn{2}{c|}{$O(7^{k/2})$ amortized} & rand. oblivious\tabularnewline
\hline
\cite{BaswanaKS12,ForsterG19} & $2k-1$ & $O(n^{1+1/k}\log n)$  & \multicolumn{2}{c|}{$O(k\log^{2}n)$ amortized} & rand. oblivious\tabularnewline
\hline
\cite{BernsteinFH19} & $2k-1$  & $\tilde{O}(n^{1+1/k})$  & \multicolumn{2}{c|}{$O(1)^{k}\log^{3}n$ worst-case} & rand. oblivious\tabularnewline
\hline
\hline
\multicolumn{6}{|l|}{\textbf{Against an adaptive adversary}}\tabularnewline
\hline
\multirow{2}{*}{\cite{AusielloFI06}} & $3$ & $O(n^{1+1/2})$  & \multicolumn{2}{c|}{$O(\Delta)$ amortized} & deterministic \tabularnewline
\cline{2-6}
 & $5$ & $O(n^{1+1/3})$  & \multicolumn{2}{c|}{$O(\Delta)$ amortized} & deterministic \tabularnewline
\hline
\multirow{2}{*}{\cite{Bernstein2020fully}} & \textbf{$\Otil(1)$ } & \textbf{$\Otil(n)$ } & \multicolumn{2}{c|}{\textbf{$\Otil(1)$ }amortized} & rand. adaptive\tabularnewline
\cline{2-6}
 & $n^{o(1)}$  & $\Otil(n)$  & \multicolumn{2}{c|}{$n^{o(1)}$ worst-case} & deterministic \tabularnewline
\hline
\multirow{3}{*}{\textbf{Ours}} & $2k-1$ & $O(n^{1+1/k}\log n)$  & $O(\log n)$ amortized & $\poly(n)$ worst-case & deterministic \tabularnewline
\cline{2-6}
 & $3$ & $\Otil(n^{1+1/2})$ & $\Otil(1)$ amortized & $\Otil(\sqrt{n})$ worst-case & rand. adaptive\tabularnewline
\cline{2-6}
 & $3$ & $O(n^{1+1/2})$ & \multicolumn{2}{c|}{$O(\min\{\Delta,\sqrt{n}\}\log n)$ worst-case} & deterministic\tabularnewline
\hline
\end{tabular}
\par\end{centering}
}
\end{adjustwidth}

\caption{\label{tab:state of the art}The state of the art of fully dynamic
spanner algorithms.}
\end{table}

\medskip\noindent
\textbf{Organization.}  In \Cref{sec:greedy}, we give a very short proof of \Cref{thm:main greedy}.
In \Cref{sec:3spanner}, we prove \Cref{thm:main proactive} assuming a crucial lemma (\Cref{th:proactive:resampling:main}) needed for bounding the recourse.
To prove this lemma, we show a new abstraction for the proactive resampling technique in \Cref{sect:job-machine} and complete the analysis in \Cref{sec:reduce to load}.
Our side result, \Cref{thm:main det worst case}, is based on the the static construction presented in \Cref{sub:sec:static} and its simple proof is given in \Cref{sub:sec:dynamic:1}.

\section{Deterministic Spanner with Near-optimal Recourse}
\label{sec:greedy}
Below, we show a decremental algorithm that \emph{handles edge deletions only} with near-optimal recourse. This will imply \Cref{thm:main greedy} by a known reduction formally stated in \Cref{lemma:fully_dyn_reduction}. To describe our decremental algorithm, let us first recall the classic greedy algorithm.

%To prove \Cref{thm:main greedy}, we first recall the classic greedy algorithm. Then, we show an algorithm that handles \emph{edge deletions only} and finally obtain \Cref{thm:main greedy} by using a known reduction \cite{BaswanaKS12}.

%In this section, we prove \Cref{thm:main greedy}.
%show a simple algorithm which maintain $(2k-1)$-spanner $H = (V,E_H)$ of $G=(V,E)$ with $O(m)$ total recourse.

\medskip\noindent
\textbf{The Greedy Algorithm.} % CITE?
Alth\"{o}fer et al.~\cite{dcg/Althofer93} showed the following algorithm for computing $(2k-1)$-spanners.
% This algorithm can be viewed as an extension of Kruskal's algorithm.
Given a graph $G=(V,E)$ with $n$ vertices, fix some \emph{order} of edges in $E$. Then, we inspect each edge one by one according to the order.
Initially $E_H = \emptyset$. When we inspect $e=(u,v)$, if $\dist_H(u,v) \geq 2k$, then add $e$ into $E_H$. Otherwise, ignore it. We have the following:
%It is known that this algorithm gives a $(2k-1)$-spanner of size at most $O(n^{1+1/k})$.

\begin{theorem}[\cite{dcg/Althofer93}]
	\label{thm:greedy classic}
    The greedy algorithm above
    %takes an input graph $G=(V,E)$ with $n$ vertices and
    outputs a $(2k-1)$-spanner $H=(V,E_H)$ of $G$ containing $O(n^{1+1/k})$ edges.
\end{theorem}

It is widely believed that the greedy algorithm is extremely bad in dynamic setting:
an edge update can drastically change the greedy spanner.
In contrary, when we allow the order in which greedy scans the edges to be changed adaptively, we
can avoid removing spanner edges until it is deleted by the adversary.
This key insight leads to optimal recourse.
When recourse is the only concern, prior to our work this result was known only for spanners
with polylog stretch, which is a much easier problem.

\noindent
\textbf{The Decremental Greedy Algorithm.}
Now we describe our deletion-only algorithm. Let $G$ be an initial graph with $m$ edges and $H = (V,E_H)$ be a $(2k-1)$-spanner with $O(n^{1+1/k})$ edges. Suppose an edge $e = (u,v)$ is deleted from the graph $G$. If $(u,v) \notin E_H$, then we do nothing.
Otherwise, we do the following. We first remove $e$ from $E_H$. Now we look at the only remaining non-spanner edges $E \setminus E_H$, one by one in an arbitrary order.
(Note that the order is \emph{adaptively} defined and not fixed through time because it is defined only on $E \setminus E_H$.)
When we inspect $(x,y) \in E \setminus E_H$, as in the greedy algorithm, we ask if $\dist_H(x,y) \geq 2k$ and add $(x,y)$ to $E_H$ if and only if it is the case. This completes the description of the algorithm.
%
%In a way, it is as if we run the greedy algorithm with the remaining edges $E$, but putting $E_H$ first in the order.

\medskip\noindent
\textbf{Analysis.}
We start with the most crucial point. We claim that the new output after removing $e$ is as if we run the greedy algorithm that first inspects edges in $E_{H}$ (the order within $E_{H}$ is preserved) and later inspects edges in $E\setminus E_{H}$.

To see the claim, we argue that if the greedy algorithm inspects $E_H$ first, then the whole set $E_{H}$ must be included, just like $E_H$ remains in the new output.
To see this, note that, for each $(x,y)\in E_{H}$, $\dist_{H}(x,y)\ge2k$ when $(x,y)$ was inspected according to some order. But, if we move the whole set $E_{H}$ to be the prefix of the order (while the order within $E_H$ is preserved), it must still be the case that $\dist_{H}(x,y)\ge2k$ when $(x,y)$ is inspected and so $e$ must be added into the spanner by the greedy algorithm.

So our algorithm indeed ``simulates'' inspecting $E_H$ first, and then it explicitly implements the greedy algorithm on the remaining part $E\setminus E_{H}$. So we conclude that it simulates the greedy algorithm. Therefore, by \Cref{thm:greedy classic}, the new output is a $(2k-1)$-spanner with $O(n^{1+1/k})$ edges.

The next important point is that, whenever an edge $e$ added into the spanner $H$, the algorithm never tries to remove $e$ from $H$. So $e$ remains in $H$ until it is deleted by the adversary. Therefore, the total recourse is $O(m)$. With this, we conclude the following key lemma:

% We claim that the algorithm above gives us the following result.
%We have the following key lemma.%claim the following result.

\begin{lemma}
    \label{thm:greedy_spanner}
    Given a graph $G$ with $n$ vertices and $m$ initial edges undergoing only edge deletions, the algorithm above maintains a $(2k-1)$-spanner $H$ of $G$ of size $O(n^{1+1/k})$ with $O(m)$ total recourse.
\end{lemma}

By plugging \Cref{thm:greedy_spanner} to the fully-dynamic-to-decremental reduction by \cite{BaswanaKS12} below, we conclude \Cref{thm:main greedy}. We also include the proof of \Cref{lemma:fully_dyn_reduction} in \Cref{sec:reduction} for completeness.

\begin{lemma}[\cite{BaswanaKS12}]
	\label{lemma:fully_dyn_reduction}
	Suppose that for a graph $G$ with $n$ vertices and $m$ initial edges undergoing only edge deletions, there is an algorithm that maintains a $(2k-1)$-spanner $H$ of size $O(S(n))$  with $O(F(m))$ total recourse where $F(m) = \Omega(m)$,
	then there exists an algorithm that maintains a $(2k-1)$-spanner $H'$ of size $O(S(n) \log n)$
	in a fully dynamic graph with $n$ vertices using  $O(F((U) \log n))$ total recourse. Here $U$ is the total number of updates, starting from an empty graph.
\end{lemma}

% \thatchaphol{STATE the reduction.}

%!TEX root=ICALP_main.tex
%\section{Efficient Dynamic Algorithms for Maintaining a $3$-Spanner}
\section{$3$-Spanner with Near-optimal Recourse and Fast Update Time}
\label{sec:3spanner}

In this section, we prove \Cref{thm:main proactive} by showing an algorithm for maintaining a $3$-spanner with small update time {\em in addition} to having small recourse. We start by explaining a basic static construction and needed data structures in \Cref{sub:sec:static} and show the dynamic algorithm in \Cref{sub:sec:dynamic:1}.
Assuming our key lemma (\Cref{th:proactive:resampling:main}) about proactive resampling, most details here are quite straight forward. Hence, due to space constraint,
most proofs are either omitted or deferred to \Cref{sec:app:missing:3spanner}.

Throughout this section,
we let $N_G(u) = \{ v \in V : (u, v) \in E \}$ denote the set of neighbors of a node $u \in V$ in a graph $G = (V, E)$, and we let $\text{deg}_G(u) = |N_G(u)|$ denote the degree of the node $u$ in the graph $G$.

\subsection{A Static Construction and Basic Data Structures}
\label{sub:sec:static}

\textbf{A Static Construction.} We now describe our static algorithm.
Though our presentation is different, our algorithm is almost identical to \cite{GrossmanP17}.
The only difference is that we do not distinguish small-degree vertices from large-degree vertices.
 %\thatchaphol{mention \cite{GrossmanP17}. Say that it is similar. Describe the difference.}\pat{Done.}

We first arbitrarily partition $V$ into $\sqrt{n}$ equal-sized \emph{buckets} $V_1, \ldots, V_{\sqrt{n}}$.
We then construct three sets of edges $E_1, E_2, E_3$.
For every bucket $V_i, i \in [1,\sqrt{n}]$, we do the following.
First, for all $v \in V \setminus V_i$, if $V_i \cap N_G(v)$ is not empty, we choose a neighbor $c_i(v) \in V_i \cap N_G(v)$ and add $(v,c_i(v))$ to $E_1$. We call $c_i(v)$ an \emph{$i$-partner} of $v$.
Next, for every edge $e=(u,v)$, where both $u,v \in V_i$, we add $e$ to $E_2$.
Lastly, for $u,u' \in V_i$ with overlapping neighborhoods, we
pick an arbitrary common neighbor $w_{uu'} \in N_G(u) \cap N_G(u')$ and add $(u,w_{uu'}),(w_{uu'},u')$ to $E_3$.
We refer to the node $w_{uu'}$ as the \emph{witness} for the pair $u,u'$.

\begin{claim}
    \label{cl:static:stretch-size}
    The subgraph $H = (V, E_1 \cup E_2 \cup E_3)$ is a $3$-spanner of $G$ consisting of at most $O(n\sqrt{n})$ edges.
\end{claim}

\noindent
\textbf{Dynamizing the Construction.} Notice that it suffices to separately maintain $E_1, E_2, E_3$, in order to maintain the above dynamic $3$-spanner.
Maintaining  $E_1$ and $E_2$ is straightforward and can be done in a fully-dynamic setting in $O(1)$ worst-case update time. Indeed, if  $e = (u,c_i(u))\in E_1$ is deleted, then we pick a new $i$-partner $c_i(u) \in V_i \cap N_G(u)$ for $u$. Maintaining $V_i \cap N_G(u)$ for all $u$ allows us to update $c_i(u)$ efficiently.
If $e=(u,u') \in E_2$, where $u,u' \in V_i$, is deleted, then we do nothing.

The remaining task, maintaining $E_3$, is the most challenging part of our dynamic algorithm.
Before we proceed, let us define a subroutine and a data structure needed to implement our algorithm.

\medskip\noindent
\textbf{Resampling Subroutine.}
We define {\sc Resample}$(u,u')$ as a subroutine that \emph{uniformly samples} a witness $w_{uu'}$ (i.e.~a common neighbor of $u$ and $u'$), if exists. Notice that, we can obtain $E_3$ by calling {\sc Resample}$(u,u')$ for all $u,u' \in V_i$ and for all $i \in [1,\sqrt{n}]$.

% At the start of the phase, we call  {\sc Resample}$(u, u')$ for all $i \in [1, \sqrt{n}]$ and $u, u' \in V_i$ with $u \neq u'$. This initializes the set $E_3$ of the phase in a randomized manner.
% Notice that, the set of witnesses $\{w_{uu'}\}$ corresponds to the set $E_3$, which we desire to maintain. Accordingly, for ease of presentation, henceforth we will only specify how our dynamic algorithm maintains these witnesses.

\medskip\noindent
\textbf{Partnership Data Structures.}
The subroutine above hints that we need a data structure for maintaining the common neighborhoods
for all pairs of vertices that are in the same bucket.
For vertices $u$ and $v$ within the same bucket, we let $P(u,v) = N_G(u) \cap N_G(v)$ be the \emph{partnership} between $u$ and $v$.
%To initiate all partnerships, for every vertex $v \in V$, for all $i \in [1,\sqrt{n}]$,we iterate all pairs $u,u' \in N_G(v) \cap V_i$ and add $v$ to $P(v,v')$. Such a data structure can be initiated in $\tilde O(n^{2.5})$.
To maintain these structures dynamically, when an edge $(u,v)$ is inserted, if $u \in V_i$ and $v \in V_j$, we add $v$ to $P(u,u')$ for all $u' \in V_i \cap N_G(v) \setminus \{u\}$, and symmetrically add $u$ to $P(v,v')$ for all $v' \in V_j \cap N_G(u) \setminus \{v\}$.
This clearly takes $O(\sqrt{n} \log{n})$ worst-case time for edge insertion. and this is symmetric for edge deletion.
%To initialize these structures from scratch on a graph with $m$, we can simply insert each edge one by one in $O(m \sqrt{n})$ time.
%Likewise, if $(u,v)$ is deleted, we can remove $v$ from $P(u,u')$ for all $u' \in V_i \cap N_G(v) \setminus \{u\}$. We update these information analogously by adding (removing) $u$ to (from) $P(v,v')$ for suitable $v'$. It is clear that the update time is $\tilde O(\sqrt{n})$ per edge insertion (deletion).

As we want to prove that our final update time is $\Otil(\sqrt{n})$, we can assume from now that $E_1,E_2$, and all partnerships are maintained in the background.

% We now analyze the amortized update time. We will need a data structure for maintaining the common neighborhoods of all distinct pairs of vertices.
% This can be initiated in $\tilde{O}(n^2)$.\thatchaphol{explain a bit more how.} For any edge $(u,v)$ that gets deleted, suppose $u \in V_i$ and $v \in V_j$, we update (1) $N_G(u)$ and $N_G(v)$, (2) we update $N_G(u) \cap N_G(n_u)$ for all $n_u \in N_G(u)$, and (3) we update $N_G(v) \cap N_G(n_v)$ for all $n_v \in N_G(v)$.  It is not hard to see that this can be done in $\tilde O(\sqrt{n})$ per edge deletion.\thatchaphol{I think we should give a short argument.}

% \subsection{Randomized $3$-Spanner: Proof of \Cref{thm:main proactive}}
\subsection{Maintaining Witnesses via Proactive Resampling}
\label{sub:sec:dynamic:1}

\textbf{Remark.} For clarity of exposition, we will present an amortized update time analysis.
Using standard approach, we can make the update time worst case. We will discuss this issue at the end of this section.

Our dynamic algorithm runs in {\em phases}, where each phase lasts for $n \sqrt{n}$ consecutive updates (edge insertions/deletions).
As a spanner is decomposable\footnote{Let $G_1 = (V,E_1)$ and $G_2 = (V,E_2)$. If $H_1$ and $H_2$ are $\alpha$-spanners $G_1$ and $G_2$ respectively, then $H_1 \cup H_2$ is a $\alpha$-spanner of $G_1 \cup G_2$.},
it suffices to maintain a $3$-spanner $H$ of the graph undergoing only edge deletions within this phase and then include all edges inserted within this phase into $H$, which increases the size of $H$ by at most $n\sqrt{n}$ edges. Henceforth, \emph{we only need to present how to initialize a phase and how to handle edge deletions within each phase.} The reason behind this reduction is because our proactive resampling technique naturally works for decremental graphs.

% Hence, it suffices to show how to maintain a $3$-spanner in a decremental graph.\thatchaphol{Where do we say that we can handle worst-case update time?}
\medskip\noindent
\textbf{Initialization.} At the start of the phase, since our partnerships structures only processed edge deletions from the previous phase,
we first update partnerships with all the $O(n\sqrt{n})$ inserted edges from the previous phase. Then, we call {\sc Resample}$(u,u')$ for all $u,u' \in V_i$ for all $i \in [1,\sqrt{n}]$ to replace all witnesses and initialize $E_3$ of this phase.

% Although the initialization time clearly exceeds our worst-case update time of $\Otil(\sqrt{n})$, with the standard building-in-the-background technique (see e.g.~\cite{thorup2005worst,baswana2016dynamic,NanongkaiSW17}), we can obtain worst-case update nonetheless. We discuss this issue at the end of the section.

% {\bf Initialization:} At the start of the phase, we run the static algorithm on the current input graph $G = (V, E)$ to obtain a $3$-spanner $H = (V, E_H)$ of $G$.\thatchaphol{Issue.}

% Recall that the edge-set $E_H$ is a union of three subsets: $E_1$, $E_2$ and $E_3$.
% We need to maintain each of these three subsets, as edges keep getting deleted from $G = (V, E)$.

\medskip\noindent
\textbf{Difficulty of Bounding Recourse.}
Maintaining $E_3$ (equivalently, the witnesses) in $\Otil(\sqrt{n})$ worst-case time is straightforward because the partnership data structure has $O(\sqrt{n} \log{n})$ update time. However, our goal is to show $\Otil(1)$ amortized recourse, which is the most challenging part of our analysis.
To see the difficulty, if $(u,v)$ is deleted and $u \in V_i$,
a vertex $v$ may serve as a witness  $\{w_{uu'}\}$ for all $u' \in V_i$.
In this case, deleting $(u,v)$ causes the algorithm to find a new witness $w_{uu'}$ for all $u' \in V_i$.
This implies a recourse of $|V_i| = \Omega(\sqrt {n})$.
% Since the adversary is adaptive, she might force this bad situation repeatedly.
To circumvent this issue, we apply the technique of {\em proactive resampling}, as described below.

\medskip
\noindent
\textbf{Proactive Resampling.} We  keep track of a  {\em time-variable} $T$; the number of updates to $G$  that have occurred in this phase until now.
$T$ is initially $0$. We increment $T$ each time an edge gets deleted from $G$.

In addition, for all $i \in [1, \sqrt{n}]$ and $u, u' \in V_i$ with $u \neq u'$, we maintain: (1) $w_{uu'}$, the {\em witness} for the pair $u$ and $u'$ and (2)
a set $\text{{\sc Schedule}}[u, u']$ of positive integers, which is the set of timesteps where our algorithm intends to {\em proactively resample} a new witness for $u,u'$. This set grows adaptively each time the adversary deletes $(u,w_{uu'})$ or $(w_{uu'}, u')$.

Finally,  to ensure that the update time of our algorithm remains small,  for each $\lambda \in [1, n \sqrt{n}]$ we maintain a $\text{{\sc List}}[\lambda]$, which consists of all those pairs of nodes $(x, x')$ such that $\lambda \in \text{{\sc Schedule}}[x, x']$.

% We now show how to handle the deletion of an edge from $G$ within the concerned phase. Suppose that, at timestep $T$, the edge getting deleted is $(u, v)$, where $u \in V_i$ and $v \in V_j$ for some $i, j \in [1, \sqrt{n}]$.    We first update the neighborhoods in the graph $G$ by setting $N_{G}(u) \leftarrow N_{G}(u) \setminus \{ v\}$ and $N_{G}(v) \leftarrow N_{G}(v) \setminus \{ u \}$.
When an edge $(u,v)$, where $u \in V_i$ and $v \in V_j$ is deleted, we do the following operations.
First, for all $u' \in V_i$ that had $v = w_{uu'}$ as a common neighbor with $u$ before deleting $(u,v)$, we add the timesteps $\{T + 2^k  \mid T + 2^k \leq n \sqrt{n}, k \in \mathbb N \}$ to $\text{{\sc Schedule}}[u,u']$.
Second, analogous to the previous one, for all $v' \in V_j$ that had $u = w_{vv'}$ as a common neighbor with $v$ before deleting $(u,v)$,
we add the timesteps $\{T + 2^k  \mid T + 2^k \leq n \sqrt{n}, k \in \mathbb N \}$ to $\text{{\sc Schedule}}[v,v']$.
Third, we set $T \leftarrow T+1$.
Lastly, for each $(x,x') \in \text{{\sc List}}[T]$, we call the subroutine
$\text{{\sc Resample}}(x,x')$.

The key lemma below summarizes a crucial property of this dynamic algorithm. Its proof appears in \Cref{sect:job-machine}.
\begin{lemma}
\label{th:proactive:resampling:main}
During a phase consisting of $n\sqrt{n}$ edge deletions, our dynamic algorithm makes at most  $\tilde{O}(\sqrt{n})$ calls to the {\sc Resample} subroutine after each edge deletion.
Moreover, the {\em total} number of calls to the {\sc Resample} subroutine during an {\em entire phase} is at most $\tilde{O}(n \sqrt{n})$ w.h.p. Both these guarantees hold  against an adaptive adversary.
%
%
%While handling a given edge deletions within a phase, w.h.p.~our dynamic algorithm makes at most $\tilde{O}(\sqrt{n})$ calls to the subroutine {\sc Resample}. Furthermore, w.h.p.~the {\em total} number of calls to the {\sc Resample} subroutine during an {\em entire phase} is at most $\tilde{O}(n \sqrt{n})$. Both these guarantees hold  against an adaptive adversary.\thatchaphol{Should make it clear that. This works under deletions only.} \pat{explicitly say that we handle deletions should solve this?}
\end{lemma}

%\begin{theorem}

\noindent
\textbf{Analysis of Recourse and Update Time.} Our analysis are given in the lemmas below.

% {\bf Recourse:}
\begin{lemma}[Recourse]
\label{lm:amortized:recourse}
The   amortized recourse of our  algorithm is $\tilde{O}(1)$ w.h.p., against  an adaptive adversary.
\end{lemma}
\begin{proof}

To maintain the edge-sets $E_1$ and $E_2$,  we pay a worst-case recourse of $O(1)$ per update. For maintaining the edge-set $E_3$, our total recourse during the entire phase is at most $O(1)$ times the number of calls made to the {\sc Resample}$(.,.)$ subroutine, which in turn is at most $\tilde{O}(n \sqrt{n})$ w.h.p.~(see Lemma~\ref{th:proactive:resampling:main}).
Finally, while computing $E_3$ in the beginning of a phase, we pay  $O(n \sqrt{n})$ recourse. Therefore, the overall total recourse during an entire phase is $\tilde{O}(n \sqrt{n})$ w.h.p.. Since a phase lasts for $n \sqrt{n}$ time steps, we conclude the lemma.
\end{proof}

% \noindent
% {\bf Amortized update time:}
% We now analyze the amortized update time in the following lemmas.
% We will need a data structure for maintaining the common neighborhoods of all distinct pairs of vertices.
% This can be initiated in $\tilde{O}(n^2)$.\thatchaphol{explain a bit more how.} For any edge $(u,v)$ that gets deleted, suppose $u \in V_i$ and $v \in V_j$, we update (1) $N_G(u)$ and $N_G(v)$, (2) we update $N_G(u) \cap N_G(n_u)$ for all $n_u \in N_G(u)$, and (3) we update $N_G(v) \cap N_G(n_v)$ for all $n_v \in N_G(v)$.  It is not hard to see that this can be done in $\tilde O(\sqrt{n})$ per edge deletion.\thatchaphol{I think we should give a short argument.}

% We continue the analysis of the update time of our dynamic algorithm.

\begin{lemma}[Worst-case Update Time within a Phase]
\label{lm:worstcase}
Within a phase,  our  algorithm handles a given update in $\tilde{O}(\sqrt{n})$ worst case time w.h.p..
\end{lemma}

\begin{proof}
Recall that the sets $E_1, E_2$ can be maintained in ${O}(1)$ worst case update time.
Henceforth, we  focus on the time required to maintain the edge-set $E_3$ %(see Figure~\ref{fig:critical})
after a given update in $G$.

 Excluding the time spent on maintaining the partnership data structure (which is $\tilde{O}(\sqrt{n})$ in the worst-case anyway),
this is proportional to $\tilde{O}(1)$ times the number of calls made to the {\sc Resample}$(.,.)$ subroutine,
{\em plus} $\tilde{O}(1)$ times the number of pairs $u,u' (v,v')$ where we need to adjust $\text{{\sc SCHEDULE}}[u,u']$.
% {\em plus} $\tilde{O}(1)$ times the number of iterations of the {\sc For} loops in steps 03 and 10 of Figure~\ref{fig:critical}.
The former is w.h.p.~at most $\tilde{O}(\sqrt{n})$ according to Lemma~\ref{th:proactive:resampling:main},
while the latter is also at most $\tilde{O}(\sqrt{n})$ since $|V_i|, |V_j| \leq \sqrt{n}$.
Thus,  within a phase we can also maintain the set $E_3$ w.h.p.~in $\tilde{O}(\sqrt{n})$ worst case update time.
\end{proof}

Although the above lemma says that we can handle each edge deletion in $\Otil(\sqrt{n})$ worst-case update time,
our current algorithm does not guarantee worst-case update time yet because the intialization time exceed the   $\Otil(\sqrt{n})$ bound. In more details,
observe that the total initialization time is $O(n\sqrt n) \times O(\sqrt{n} \log{n})$ because we need to insert  $O(n\sqrt{n})$ edges into partnership data structures, which has $O(\sqrt{n} \log{n})$ update time.  Over a phase of $n\sqrt{n}$ steps, this implies only  $\Otil(\sqrt{n})$ amortized update time.

However, since the algorithm takes long time only at the initialization of the phase, but takes $\Otil(\sqrt{n})$ worst-case step for each update during the phase, we can apply the standard building-in-the-background technique (see~\Cref{sub:app:spread:work}) to de-amortized the update time.
We conclude the following:
%\thatchaphol{please take care of this details in appendix.}

 %Finally,  the time taken to compute $E_3$ at the start of a phase is $\tilde{O}(n^2)$, which gets amortized over the length of the phase. As a phase lasts for $n \sqrt{n}$ updates, Lemma~\ref{lm:worstcase}  implies that the overall amortized update time of our  algorithm is $\tilde{O}(\sqrt{n}) + \frac{\tilde{O}(n^2)}{n \sqrt{n}} = \tilde{O}(\sqrt{n})$ w.h.p. This is summarized in the lemma below.

\begin{lemma}[Worst-case Update Time for the Whole Update Sequence]
\label{lm:worstcase:updatetime}
W.h.p., the worst-case update time of our dynamic algorithm is $\tilde{O}(\sqrt{n})$.
\end{lemma}

\section{Proactive Resampling: Abstraction}
\label{sect:job-machine}

The goal of this section is to prove \Cref{th:proactive:resampling:main} for bounding the recourse of our 3-spanner algorithm. This is the most technical part of this paper.
To ease the analysis, we will abstract the problem situation in \Cref{th:proactive:resampling:main} as a particular dynamic problem of assigning jobs to machines while an adversary keeps deleting machines and the goal is to minimize the total number of reassignments. Below, we formalize this problem and show how to use it to bound the recourse of our 3-spanner algorithm.

Our abstraction has two technical contributions:
(1) it allows us to easily work with multiple sampling probabilities, while in  \cite{Bernstein2020fully}, they fixed a single parameter on sampling probability,
(2) the simplicity of this abstraction can expose the generality of the proactive resampling technique itself; it is not specific to the cut sparsifier problem as used in \cite{Bernstein2020fully}.

%We believe that our abstraction is generic and might find future applications. It can be viewed as an exposition to the proactive resampling technique \cite{Bernstein2020fully}.\thatchaphol{change wording}
\medskip\noindent
\textbf{Jobs, Machines, Routines, Assignments, and Loads.}
Let $J$ denote a set of \emph{jobs} and $M$ denote a set of \emph{machines}.
We think of them as two sets of vertices of the (hyper)-graph $G=(J,M,R)$.\footnote{This graph is different from the graph that we maintain a spanner in previous sections.}
A \emph{routine} $r \in R$ is a hyperedge of $G$ such that $r$ contains exactly one job-vertex from $J$, denoted by $\job(r)\in J$, and may contain several machine-vertices from $M$, denoted by $M(r) \subseteq M$.
Each routine $r$ in $G$ means there is a routine for handling $\job(r)$ by \emph{simultaneously} calling machines in $M(r)$.
Note that $r = \{\job(r)\}\cup M(r)$.
We say that $r$ is a \emph{routine for} $\job(r)$.
For each machine $x \in M(r)$, we say that routine $r$ \emph{involves} machine $x$, or that $r$ \emph{contains} $x$. The set $R(x)$ is then defined as the set of routines involving machine $x$.
Observe that there are $\deg_G(u)$ different routines for handling job $u$.
%
%Consider a hypergraph $G=(J,M,R)$. We can think of $u \in J$ as a job and $x \in M$ as a machine. $R$ is a set of hyperedges in $G$.
%We usually use $u,v,w$ to denote jobs, i.e., vertices in $J$ and use $x,y,z$ to denote machines, i.e., vertices in $M$.
%Each hyper edge $e \in R$ in our graph has exactly one node in $J$.
%We let $\job(e) = J \cap e$ be the job that is contained in $e$.
%$M(e) \subseteq M$ is the set of machines required to handle $\job(e)$.
%In this perspective, for any job $u \in J$, there are $\deg(u)$ ways of handling it.
An \emph{assignment} $A = (J,M, F \subseteq R)$ is simply a subgraph of $G$. We say assignment $A$ is \emph{feasible} iff $\deg_A(u) = 1$ for every job $u \in J$ where $\deg_G(u)>0$. That is, every job is handled by some routine, if exists. When $r \in A$, we say that $\job(r)$ is \emph{handled by} routine $r$.
Finally, given an assignment $A$, the \emph{load} of a machine $x$ is the number of routines in $A$ involving $x$, or in other words, is the degree of $x$ in $A$, $\deg_A(x)$. 
% PAT: Check with others if this is okay.
We note explicitly that our end-goal is not to optimize loads of machines. 
Rather, we want to minimize the number of reassignments needed to maintain feasible assignments throughout the process.

In this section, we usually use $u,v,w$ to denote jobs, use $x,y,z$ to denote machines, and use $e$ or $r$ to denote routines or (hyper)edges.

%Let now define desired properties of an assignment.
%
%\begin{definition}
%    For a given graph $G$, we call an assignment $A \subseteq G$ \emph{feasible} iff $\deg_A(u) \geq 1$ for every job $u \in J$.
%\end{definition}
\medskip\noindent
\textbf{The Dynamic Problem.}
Our problem is to maintain a feasible assignment $A$ while the graph $G$ undergoes a sequence of machine deletions (which might stop before all machines are deleted).
More specifically, when a machine $x$ is deleted, all routines $r$ containing $x$ are deleted as well.
But when routines in $A$ are deleted,
$A$ might not be feasible anymore and we need to add new edges to $A$ to make $A$ feasible. Our goal is to minimize the total number of routines ever added to $A$.

To be more precise, write the graph $G$ and the assignment $A$ after $t$ machine-deletions as $G^t = (J,M,R^t)$ and $A^t =(J,M,F^t)$, respectively.
%So $G^0 = G$ is the initial graph and $G^t$ is the graph after $t$ edges are deleted.
Here, we define \emph{recourse} at timestep $t$ to be $|F^t \setminus F^{t-1}|$, which is the number of routined added into $A$ at timestep $t$.
When the adversary deletes $T$ machines, the goal is then to minimize the total recourse $\sum_{t = 1}^{T} |F^t \setminus F^{t-1}|$.

\medskip\noindent
\textbf{The Algorithm: Proactive Resampling.}
To describe our algorithm, first let $\Resample(u)$ denote the process of reassigning job $u$ to a uniformly random routine for $u$. In the graph language, $\Resample(u)$ removes the edge $r$ such that $\job(r) = u$ from $A$, sample an edge $r'$ from $\{r \in R~|~\job(r) = u\}$, and then add $r'$ into $A$.
At timestep $0$, we initialize a feasible assignment $A^0$ by calling $\Resample(u )$ for every job $u \in J$, i.e., assign each job $u$ to a random routine for $u$. Below, we describe how to handling deletions.

Let $T$ be the total number of machine-deletions. For each job $u$, we maintain $\Schedule(u) \subseteq [T]$ containing all time steps that we will invoke $\Resample(u)$.
That is, at any timestep $t$, before an adversary takes any action, we call $\Resample(u)$ if $t \in \Schedule(u)$.

We say that an adversary \emph{touches} $u$ at timestep $t$ if the routine $r \in A^t$ handling $u$ at time $t$ is deleted. When $u$ is touched, we call $\Resample(u)$ and, very importantly, we put $t+1, t+2, t+4, \ldots$ where $t+2^i \le T$ into $\Schedule(u)$.
This is the action that we call \emph{proactive resampling} because we do not just resample a routine for $u$ only when $u$ is touched, but do so proactively in the future as well. This completes the description of the algorithm.

\begin{comment}
Now we describe what we mean by proactive resampling.
For each vertex $u$, we maintain $\Schedule(u)$, which contains all time steps in the future that we will invoke $\Resample(u)$. We proactively put timesteps where we would call $\Resample(u)$ into $\Schedule(u)$, hence the term \emph{proactive resampling}.

%Initially, at timestep $0$, we put $1, 2, 4, 8, \ldots, 2^i, \ldots$, in the schedule.
We say that an adversary \emph{touches} $u$ at timestep $t$ if an edge $e \in A^t$ with $\job(e) = u$ is deleted at timestep $t$. When $u$ is touched, we call $\Resample(u)$ and we put $t+1, t+2, t+4, \ldots$ into $\Schedule(u)$.

At any timestep $t'$, before an adversary takes any action, we check $\Schedule(u)$ and
see if $t' \in \Schedule(u)$.
If so, then we call $\Resample(u)$.

\end{comment}

Clearly, $A$ remains a feasible assignment throughout because whenever a job $u$ is touched, we immediately call $\Resample(u)$. The key lemma below states that the algorithm has low recourse, even if the adversary is adaptive in the sense that each deletion at time $t$ depends on previous assignment before time $t$.

\begin{lemma}
\label{lemma:second_guarantee}
Let $T$ be the total number of machine-deletions.
The total recourse of the algorithm running against an adaptive adversary
is $O\big( |J|\log(\Delta)\log^2|M| + T\log^2 |M|\big)$ with high probability
where $\Delta$ is the maximum degree of any job.
Moreover, if the load of a machine never exceeds $\lambda$,
then our algorithm has $O(\lambda \log{T})$ worst-case recourse.
\end{lemma}

\begin{comment}
\begin{lemma}
	\label{lemma:second_guarantee}
    With high probability, throughout $T$ time steps, for every time step,
    this algorithm maintains a feasible assignment $A$ against an adaptive adversary.
     % with overhead $\left (O \left (\log{(T)} \right), O(\log |M|) \right)$.
    The total recourse of the algorithm is bounded above by $O\big( |J|\log^2|M| + T\log |M|\big)$.
    % The recourse of the algorithm
    % is bounded above by $O\left( \max_e \target(e) ( \log T )  + \log |M| \right)$ for each timestep.
    % Also, the total number of $\Resample_\ind$ calls per each timestep is bounded above by $O(\log T)$.

    % Moreover, if we let $\Delta = \displaystyle \max_{u \in J, t \in T} \deg_{A^t}(u)$, then $\Delta = O(\log {|J|})$ with high probability.
    % Here, $\Delta_J = \displaystyle \max_{u \in J} \deg_G(u)$.
\end{lemma}
\end{comment}

% The rest of this section is devoted to proving \Cref{lemma:second_guarantee}.
We will prove \Cref{lemma:second_guarantee} in \Cref{sec:reduce to load}.
 Before proceeding any further, however, we argue why \Cref{lemma:second_guarantee} directly bounds the recourse of our 3-spanner algorithm.

\medskip\noindent
\textbf{Back to 3-spanners: Proof of  \Cref{th:proactive:resampling:main}.}
It is easy to see that maintaining $E_3$ in our $3$-spanner algorithm can be framed exactly as the job-machine load-balancing problem.
Suppose the given graph is $G = (V,E)$ where $n = |V|$ and $m = |E|$.
We create a job $j_{uu'}$ for each pair of vertices $u,u' \in V_i$ with $u \neq u'$.
% For brevity, we will say that a pair of vertices $u,u' \in V_i$ is a \emph{pairing}.
For each edge $e \in E$, we create a machine $x_e$.
Hence, $|J| = O(n^{1.5})$ and $|M| = |E| = m$.
For each job, as we want to have a witness $w_{uu'}$, this witness is corresponding
to two edges $e=(u,w_{uu'})$ and $e' = (u',w_{uu'})$. Hence, we create a routine $r = (j_{uu'}, e, e')$ for each $u,u' \in V_i$ and a common neighbor $w_{uu'}$.
Since there are at most $n$ common neighbors between each $u$ and $u'$, $\Delta = O(n)$.
% The maximum number of routines $|R|$ is then bounded by $O(n^{2.5})$, but should be much less in general.
A feasible assignment is then corresponding to finding a witness for each job.
% \thatchaphol{Describe the connection. What is job,machine,routine,assignment.}
Our algorithm that maintains the spanner is also exactly this load-balancing algorithm. Hence,
the recourse of the $3$-spanner construction follows from Lemma~\ref{lemma:second_guarantee} where we delete exactly
$T=O(n^{1.5})$ machines. As $|J| = O(n^{1.5})$, the total recourse bound then becomes $O(n^{1.5} \log^3 n)$.
As $T=O(n^{1.5})$, averaging this bound over all timesteps yields $O(\log^3 n )$ amortized recourse.
\section{Proactive Resampling: Analysis (Proof of \Cref{lemma:second_guarantee})}
\label{sec:reduce to load}

The first step to prove \Cref{lemma:second_guarantee} is to bound the loads of machines $x$. This is because whenever machine $x$ is deleted, its load of $\deg_A(x)$ would contribute to the total recourse.

What would be the expected load of each machine?
For intuition, suppose that the adversary was \emph{oblivious}. Recall that $R(x)$ denote the set of all routines involving machine $x$.
Then, the expected load of machine $x$ would be $\sum_{r\in R(x)} 1 / \deg_G (\job(r))$ because each job samples its routine uniformly at random, and this is concentrated with high probability using  Chernoff's bound.
%So even when all machines are deleted, roughly speaking, the contribution to the total recourse would be $\approx \sum_{x\in M} \sum_{e \ni x} 1 / \deg_G (\job(e)) = \sum_{e\in R} 1 / \deg_G (\job(e)) = |J|$, which is stronger what \Cref{lemma:second_guarantee} requires.\footnote{The precise total recourse in the oblivious setting is $O(|J|\log\Delta)$, but we will not need this.}
%
Although in reality our adversary is \emph{adaptive}, our plan is to still prove that the loads of machines do not exceed its expectation in the oblivious setting too much. This motivates the following definitions.

%To bound the recourse, we look at the \textbf{load} of a machine $x$, i.e., the number of edges
%adjacent to $x$ in $A$.
%We want $\deg_A(x)$ to be small. How small can we expect $\deg_A(x)$ to be?
%In this perspective, we can first see what we can do in the \emph{fractional} setting
%where we can select edges adjacent to job $u$ fractionally.
%For $\deg_A(u)$ to be $1$, one solution is to assign $1/\deg_A(u)$ to each edge adjacent to $u$.
%Summing this up from a machine's perspective, we get the target load for that machine.

\begin{definition}
	The \emph{target load} of machine $x$ is $\target(x) = \sum_{r \in R(x)} 1 / \deg_G (\job(r))$. The target load of $x$ at time $t$ is $\target^t(x) = \sum_{r  \in R^t(x)} 1 / \deg_{G^t} (\job(r))$.
    An assignment $A$ has \emph{overhead} $(\alpha, \beta)$ iff $  \deg_A(x) \leq \alpha \cdot \target(x) + \beta$
    for every machine $x \in M$.
\end{definition}

Our key technical lemma is to show that, via proactive resampling, the loads of machines are indeed close to its expectation in the oblivious setting. That is, the maintained assignment has small overhead. 
Recall that $T$ is the total number of machine-deletions.

%This bound is obvious if the adversary is an oblivious adversary. In that case, the adversary has no way to know the outcome of the random process and cannot accumulate load on one machine. Per job, we will have to fix its routine $O(\log |M|)$ times in average.

\begin{lemma}
	\label{lemma:overhead}
	With high probability, the assignment $A$ maintained by our algorithm always has overhead $\left (O(\log{T}), O(\log |M|) \right)$  even when the adversary is adaptive.
%	Let $T$ be the total number of machine-deletions.
%    With high probability, throughout $T$ time steps, for every time step,
%    this algorithm maintains a feasible assignment $A$ against an adaptive adversary
%	with overhead $\left (O \left (\log{(T)} \right), O(\log |M|) \right)$.
\end{lemma}

Due to space limit, the proof of \Cref{lemma:overhead} will given in \Cref{sec:proof_overhead}.
% Before proving \Cref{lemma:overhead},
Assuming \Cref{lemma:overhead}, we formally show how to bound of machine loads implies the total recourse, which proves \Cref{lemma:second_guarantee}.

\begin{proof}[Proof of \Cref{lemma:second_guarantee}]
%	Assuming Lemma~\ref{lemma:overhead}, we will prove Lemma~\ref{lemma:second_guarantee}.

Let $T$ be the total number of deletions.
Observe that the total recourse up to time $T$ is precisely the total number of $\Resample(.)$ calls up to time $T$, which in turn is at most the total number of $\Resample(.)$ calls put into $\Schedule(.)$ since time $1$ until time $T$.
Therefore, our strategy is to bound, for each time $t$, the number of $\Resample(.)$ calls \emph{newly generated} at time $t$.
Let $x^t$ be the machine deleted at time $t$. Observe this is at most $O(\log T) \times \deg_{A^t}(x^t)$ where $\deg_{A^t}(x^t)$ is $x^t$'s load at time $t$ and the $O(\log T)$ factor is because of proactive sampling.

By \Cref{lemma:overhead}, we have $\deg_{A^t}(x^t) \le O\left(\log{(T)}\target^t(x^t) + \log {|M|}\right)$. Also, we claim that $\sum_{t=1}^T \target^t(x^t) = O(|J| \log {\Delta})$ where $\Delta$ is the maximum degree of jobs (to be proven below).
Therefore, the total recourse up to time $T$ is at most
\begin{align*}
	O(\log T)\sum_{t=1}^{T}\deg_{A^{t}}(x^{t})
	& \le O(\log T)\sum_{t=1}^{T}O\left(\log(T)\target^{t}(x^{t})+\log|M|\right)
	\\&\le O\left(|J|\log{(\Delta)}\log^{2}{|M|}+T\log^{2}{|M|}\right)
\end{align*}
as $T \le |M|$.

%	It it straightforward from our algorithm that the total recourse is proportional to the number of $\Resample$ calls.
%	The numbers of $\Resample$ calls newly generated at time $t$ is then
%	proportional to the load of the machine that is getting deleted at this timestep.
%	Precisely, when an edge is touched, we generate $\log T$ $\Resample$ calls in $\log T$ different
%    timesteps in the future.
%	The number of edges touched by a machine deletion is exactly the load of that machine.
%	By Lemma~\ref{lemma:overhead}, we know that this load is close to
%	the expected target load of the machine, upto $\log T$ multiplicative factor.
%
%	Let us bound the number of edges that are touched throughout the algorithm.
%	At time $t$, this number is bounded by $O(\log{(T)}\target^t(x_t) + \log {|M|})$ where $x_t$ is the machine deleted
%	at time $t$. Summing this up over $T$ time steps, we get that the number of total touched edges is
%	$$ \sum_{t=1}^T O(\log{(T)} \target^t(x_t) + \log{|M|}) = O(\log{(T)} \sum_{t=1}^T \target^t(x_t) + T\log {|M|}).$$
%	If $\sum_{t=1}^T \target^t(x_t) = O(|J| \log {|M|})$, then we have that the total recourse is
%	$$ O(|J| \log {|M|} \log^2 {(T)}  + T\log {|M|}) = O(|J| \log^3 {|M|} + T \log^2 {|M|}).$$
%
It remains to show that $\sum_{t=1}^T \target^t(x^t) = O(|J| \log \Delta)$.
%
%	First note that $\sum_{x \in M} \target^t(x) \leq |J|$ for any time $t$.
%	This is because each job contributes exactly one load to machines.
%	In this way,
Recall that $\target^t(x) = \sum_{r \ni x} \frac{1}{\deg_{G^t} (\job(r))}$.
Imagine when machine $x^t$ is deleted at time $t$. We will  show how to charge $\target^t(x^t)$ to jobs with edges connecting to $x^t$.
%	to $u$ and $x_t$
For each job $u$ with $c$ (hyper)edges connecting to $x^t$, $u$'s contribution of $\target^t(x^t)$ is  $c/\deg_{G^t}(u)$.
%
%	When a machine $x_t$ is deleted at time $t$, for every edge $e$ adjacent
%	to $u$ and $x_t$, we can distribute the charge of
%	$1/\deg_{G^t}(u)$ back to job $u$ adjacent to $x_t$.
%	If $c$ is the number of edges adjacent to both $u$ and $x_t$,
%	then
	So we distribute
	the charge of $c/\deg_{G^t}(u)
	\leq \frac{1}{\deg_{G^t}(u)} +\frac{1}{\deg_{G^t}(u) - 1} + \ldots + \frac{1}{\deg_{G^t}(u)-c+1}$ to $u$.
	Since these edges are charged from machine $x^t$ to job $u$ only once, the total charge of each job $u$ at most
	$\frac{1}{\deg_G(u)}+ \frac{1}{\deg_G(u)-1} + \dots + 1/2 + 1 = O(\log \Delta).$
	Since there are $|J|$ jobs, the bound $\sum_{t=1}^T \target^t(x^t) = O(|J| \log \Delta)$ follows.

    To see that we have worst-case recourse, one can look at any timestep $t$.
    There are $O(\log{t})$ timesteps that can cause $\Resample$ to be invoked at timestep $t$,
    namely, $t-1, t-2, t-4, \ldots$. At each of these timesteps $t'$, one machine is deleted, so
    the number of $\Resample$ calls added from timestep $t'$ is also bounded by
    the load of the deleted machine $x_{t'}$, which does not exceed $\lambda$.
    Summing this up, the number of calls we make at timestep $t$ is at most $O(\lambda \log{t}) = O(\lambda \log{T})$.
	This concludes our proof.
\end{proof}

\section{Conclusion}

In this paper, we study fully dynamic spanner algorithms against an
adaptive adversary. Our algorithm in \Cref{thm:main greedy} maintains
a spanner with near-optimal stretch-size trade-off using only $O(\log n)$
amortized recourse. This closes the current oblivious-vs-adaptive gap with
respect to amortized recourse. Whether the gap can be closed for \emph{worst-case
recourse} is an interesting open problem.

The ultimate goal is to show algorithms against an adaptive adversary
with polylogarithmic amortized update time or even worst-case. Via
the multiplicative weight update framework \cite{Fleischer00,GargK07}, such algorithms would
imply $O(k)$-approximate multi-commodity flow algorithm with $\Otil(n^{2+1/k})$
time which would in turn improve the state-of-the-art.
We made partial
progress toward this goal by showing the first dynamic 3-spanner algorithms
against an adaptive adversary with $\Otil(\sqrt{n})$ update time
in \Cref{thm:main det worst case} and \emph{simultaneously} with
$\Otil(1)$ amortized recourse in \Cref{thm:main proactive}, improving
upon the $O(n)$ amortized update time since the 15-year-old work by \cite{AusielloFI06}.

Generalizing our \Cref{thm:main det worst case} to dynamic $(2k-1)$-spanners of size $\Otil(n^{1+1/k})$, for any $k \ge 2$, is also a very interesting and challenging open question.

\begin{comment}
\begin{itemize}
\item Future:
\begin{itemize}
\item Oblivious-vs-adaptive gap
\begin{itemize}
\item Greedy: shave $O(\log n)$ factor to $O(1)$
\item Worst-case recourse??
\item Update time:
\begin{itemize}
\item Application to multicommodity flow: $O(k)$-approx in $O(m+n^{1+1/k})$
time.
\end{itemize}
\end{itemize}
\item $(2k-1)$-stretch with $\Otil(n^{1+1/k})$-size on weight graphs,
light spanners
\begin{itemize}
\item Is this known even for oblivious and recourse!
\end{itemize}
\end{itemize}
\end{itemize}
\end{comment}

\appendix

\bibliography{citations}

%!TEX root=main_writeup.tex
%\subsection{Bounding Load}
%\label{sec:bound load}
\section{Proof of \Cref{lemma:overhead}}
\label{sec:proof_overhead}

% The rest of this section is devoted to the proof of \Cref{lemma:overhead}.

Here, we show that the load $\deg_{A^t}(x)$ of every machine $x$ at each time $t$ is small.
Some basic notions are needed in the analysis.

\medskip\noindent
\textbf{Experiments and Relevant Experiments.}
An \emph{experiment} $X$ is a binary random variable associated with an edge/routine $e$ and time step $t$, where $X = 1$ iff $\Sampling(\job(e))$ is called at time $t$ and $e$ is chosen to handle $\job(e)$, among all edges incident to $\job(e)$. Observe that $\P[X=1] = 1/\deg_{G^t}(\job(e))$.
Note that each call to $\Sampling(u)$ at time $t$ creates new $\deg_{G^t}(u)$ experiments. We order all experiments $X_1,X_2,X_3,\dots$ by the time of their creation.
For convenience, for each experiment $X$, we let $e(X)$, $t(X)$, and $\job(X)$ denote its edge, time of creation, and job respectively.

Next, we define the most important notion in the whole analysis.
\begin{definition}\label{def:rel}
	For any time $t$ and edge $e \in R^t$ at time $t$,
	an experiment $X$ is \emph{$(t,e)$-relevant} if
	\begin{itemize}
		\item $e(X) = e$, and
		\item there is no $t(X) < t' < t$ such that $t' \in \Schedule^{t(X)}(\job(e))$.
	\end{itemize}
	Moreover, $X$ is $(t,x)$-relevant if it is $(t,e)$-relevant and edge $e \in R^{t}(x)$ is incident to $x$.
\end{definition}
Intuitively, $X$ is a $(t,e)$-relevant experiment if $X$ could cause $e$ to appear in the assignment $A^t$ at time $t$.
To see why, clearly if $e(X) \neq e$, then $X$ cannot cause $e$ to appear. Otherwise, if $e(X) = e$ but there is $t' \in (t(X),t)$ where $t' \in \Schedule^{t(X)}(\job(e))$, then $X$ cannot cause $e$ to appear at time $t$ either.
This is because even if $X$ is successful and so $e$ appears at time $t(X)$, then later at time $t' > t(X)$, $e$ will be resampled again, and so $X$ has nothing to do whether $e$ appears at time $t > t'$.
With the same intuition, $X$ is $(t,x)$-relevant if $X$ could contribute to the load $\deg_{A^t}(x)$ of machine $x$ at time $t$.

It is important to note that, we decide whether $X$ is a $(t,e)$-relevant based on $\Schedule^{t(X)}(\job(e))$ at time $t(X)$. If it was based on $\Schedule^{t}(\job(e))$ at time $t$, then there would be only a single experiment $X$ that is $(t,e)$-relevant (which is the one with $e(X)=e$ and maximum $t(X)<t$).

According to \Cref{def:rel}, there could be more than one experiments that are $(t,e)$-relevant. For example, suppose $X$ is $(t,e)$-relevant. At time $t(X) + 1$, the adversary could touch $\job(e)$, hence, adding $t(X) + 2, t(X) + 4, \ldots$ into $\Schedule(\job(e))$. Because of this action, there is another experiment $X'$ that is $(t,e)$-relevant and $t(X') > t(X)$. This motivates the following definition.

\begin{definition}
	Let $Rel(t,e)$ be the random variable denoting the number of $(t,e)$-relevant experiments,
	%Note that there is at most one $(t,e)$-relevant experiment per timestep.
	% Let $Rel(t,e)$ be the random variable denotes the number of timestep $t' < t$ such that $t'$ is $(t,e)$-relevant.
	and let  $Rel(t,x) = \sum_{e \in R^t(x)} Rel(t,e)$ denote the total number of $(t,x)$-relevant experiments.
\end{definition}

To simplify the notations in the proof of \Cref{lemma:overhead} below, we assume the following.
\begin{assumption}[The Machine-disjoint Assumption]\label{assump:disjoint}
	For any routines $e,e'$
	with $\job(e) = \job(e')$, then $M(e) \cap M(e') = \emptyset$. That is, the edges adjacent to the same job are machine-disjoint.
\end{assumption}
Note that this assumption indeed holds for our $3$-spanner application. This is because any two paths of length $2$ between a pair of centers $u$ and $u'$ must be edge disjoint in any simple graph.
We show how remove this assumption in \Cref{sec:lift_assumption}, but the notations are more complicated.

\medskip\noindent
\textbf{Roadmap for Bounding Loads.}
We are now ready to describe the key steps for bounding the load $\deg_{A^{t}}(x)$, for any time $t$ and machine $x$.

First, we write $\mathcal X^{(t,x)} = X_1^{(t,x)}, X_2^{(t,x)}, \ldots, X_{Rel(t,x)}^{(t,x)}$ as the sequence of all $(t, x)$-relevant experiments (ordered by time step the experiments are taken).
The order in $\mathcal X^{(t,x)}$ will be important only later. For now,
we write
%observe that $\mathcal X^{(t,x)}$ lists all experiments that could contribute to the load $\deg_{A^t}(x)$  of $x$ at time $t$. Therefore, if we let
$$\overline{\deg}_{A^t}(x) = \sum_{X \in \mathcal X^{(t,x)}} X,$$
as the total number of success $(t, x)$-relevant experiments.
As any edge $e$ adjacent to $x$ in $A^t$ may appear only because of some successful $(t,x)$-relevant experiment $X \in \mathcal  X^{(t,x)}$, we conclude the following:
%then $\overline{\deg}_{A^t}(x)$ should be an upper bound of $\deg_{A^t}(x)$. This is summarized formally below.
\begin{lemma}[Key Step 1]
	\label{lem:bound deg}
	$\deg_{A^{t}}(x)\le\overline{\deg}_{A^{t}}(x)$.
\end{lemma}
%\begin{proof}
%	Each edge $e$ adjacent to $x$ in $A^t$ appears because of some successful $(t,x)$-relevant experiment $X_e$.
%	This experiment must appear in $\mathcal  X^{(t,x)}$. As $\mathcal X^{(t,x)}$ also includes other unsuccessful experiments and
%	successful experiments that got resampled, the inequality holds.
%	% By \Cref{claim:log_relevant}, there are at most $\log{T}$ relevant experiments per an edge, so we never over-count successful experiments more than $\log{T}$ times. Hence, the right inequality holds.
%\end{proof}

\Cref{lem:bound deg} reduces the problem to bounding $\overline{\deg}_{A^{t}}(x)$.
If all $(t, x)$-relevant experiments $\mathcal X^{(t,x)}= \{X_i^{(t,x)}\}_i$ were independent, then we could have easily applied standard concentration bounds to $\overline{\deg}_{A^{t}}(x) = \sum_{X \in \mathcal X^{(t,x)}} X$. Unfortunately, they are not independent as the outcome of earlier experiments can affect
the adversary's actions, which in turn affect later experiments.

Our strategy is to relate the sequence $\mathcal X^{(t,x)}$ of $(t, x)$-relevant experiments to another  sequence  $\hat{\mathcal X}^{(t,x)} = \hat{X}_1^{(t,x)}, \hat{X}_2^{(t,x)} \ldots, \hat{X}_{Rel(t,x)}^{(t,x)}$ of \emph{independent} random variables defined as follows.
For each $(t, x)$-relevant experiment $\hat{X}_i^{(t,x)}$ where  $e = e(\hat{X}_i^{(t,x)})$ and $u = \job(e)$,
we carefully define $\hat{X}_i^{(t,x)}$ as an \emph{independent} binary random variable such that $$\P[\hat{X}_i^{(t,x)} = 1 ]= 1 /\deg_{G^t}(u),$$ which is the probability that $\Sampling(u)$ chooses $e$ at time $t$. We similarly define
$$\widehat{\deg}_{A^t}(x) = \sum_{\hat{X} \in \hat{\mathcal X}^{(t,x)}} \hat X,$$
that sums independent random variables, where each term in the sum is closely related to the corresponding $(t, x)$-relevant experiments.
By our careful choice of $\P[\hat{X}_i^{(t,x)} = 1 ]$, we can
relates $\widehat{\deg}_{A^t}(x)$ to $\overline{\deg}_{A^t}(x)$ via the notion of \emph{stochastic dominance} defined below.

\begin{definition}
	Let $Y$ and $Z$ be two random variables not necessarily defined on the same probability space.
	We say that $Z$ \emph{stochastically dominates} $Y$, written as $Y \preceq Z$, if for all $\lambda \in \mathbb R$, we have $\P[Y \geq \lambda] \leq \P[Z \geq \lambda]$.
\end{definition}

Our second important step is to prove the following:
\begin{lemma}[Key Step 2]
	\label{lem:bound deg bar}
	$\overline{\deg}_{A^t}(x) \preceq \widehat{\deg}_{A^t}(x)$.
\end{lemma}

\Cref{lem:bound deg bar}, which will be proven in \Cref{sub:key 2}, reduces the problem to bounding $\widehat{\deg}_{A^{t}}(x)$, which is indeed relatively easy to bound because it is a sum of independent random variables.
The last key step of our proof does exactly this:

\begin{lemma}[Key Step 3]
	\label{lem:bound deg hat}
	$ \widehat{\deg}_{A^t}(x) \le 2 \log{(t)} \target^t(x) + O(\log{|M|})$ with probability $1-1/{|M|}^{10}$.
\end{lemma}

We prove \Cref{lem:bound deg hat}  in \Cref{sub:key 3}. Here, we only mention one important point about the proof.
The $\log(t)$ factor above follows from the factor the number of $(t,e)$-relevant experiment is always at most $Rel(t,e) \le \log(t)$ for any time $t$ and edge $e$. This property is so crucial and, actually, is what the proactive resampling technique is designed for.

Given three key steps above (\Cref{lem:bound deg,lem:bound deg bar,lem:bound deg hat}), we can conclude the proof of \Cref{lemma:overhead}.

\medskip\noindent
\textbf{Proof of \Cref{lemma:overhead}.} Recall that we ultimately want to show that,
for every timestep $t$, the maintained assignment $A^t$ has overhead $O(\log{(T)},\log|M|)$.
In other words, for every $t \in T$ and every $x \in M$,
we want to show that
$$\deg_{A^t}(x) \leq \target^t(x) \cdot O(\log{(T)}) + O(\log |M|).$$
By \Cref{lem:bound deg}, it suffices to show that
$$\overline{\deg}_{A^t}(x) \leq \target^t(x) \cdot O(\log{(T)}) + O(\log |M|).$$
By \Cref{lem:bound deg bar,lem:bound deg hat},
\begin{align*}
&~\P[\overline{\deg}_{A^t}(x) \geq 2 \log(t) \target^t(x) + O(\log |M|) ]\\
\leq&~\P[\widehat{\deg}_{A^t}(x) \geq 2 \log(t) \target^t(x) + O(\log |M|)] & \text{~(\Cref{lem:bound deg bar})} \\
\leq&~1/{|M|}^{10}. & \text{~(\Cref{lem:bound deg hat})}
\end{align*}

Now we apply union bound to the probability above. There are $T\leq |M|$ timesteps
and $|M|$ machines, hence the probability that a bad event happens
is bounded by $ \frac{T|M| }{{|M|}^{10}} = \frac{1}{{|M|}^8}.$
Here, we conclude the proof of \Cref{lemma:overhead}.

\subsection{Key Step 2}
\label{sub:key 2}

The goal of this subsection is to prove \Cref{lem:bound deg bar}.
The following reduces our problem
into proving that a certain probabilistic condition holds.

\begin{lemma}[Lemma 1.8.7(a) \cite{Doerr2020}]
	\label{lem:doerr}
	Let $X_1, \ldots, X_n$ be arbitrary boolean random variables and let
	$X^*_1 \ldots X^*_n$ be independent binary random variables.
	If we have
	$$\P[X_i = 1| X_1 = x_1, \ldots, X_{i-1} = x_{i-1}] \leq \P[X^*_i = 1]$$
	for all $i \in [n]$ and all $x_1, \ldots, x_{i-1} \in \{0,1\}$ with
	$\P[X_1 = x_1, \ldots, X_{i-1} = x_{i-1}] >0$, then
	$$ \sum_{i=1}^n X_i \preceq \sum_{i=1}^n X^*_i.$$
\end{lemma}
% \begin{observation}
%     $\mathcal X^{(t,e)} \preceq \mathcal Y^{(t,e)}$. In words, $\mathcal Y^{(t,e)}$ unconditionally sequentially dominates $\mathcal X^{(t,e)}$.
% \end{observation}

%This is exactly why we define $\widehat{\deg}(x)$.

%\begin{proof}[Proof of~\Cref{lem:bound deg bar}]
%To show that $\overline{\deg}(x) \preceq \widehat{\deg}(x)$,
%we need to show that the condition in \Cref{lem:doerr} holds when comparing ${\mathcal X}^{(t,x)}$ and $\hat{\mathcal X}^{(t,x)}$.
%Here is why we define them to be sequences, rather than their summations.
%Proving this can be done in two steps.
In light of the above lemma, we will prove that
	$$\P[X_i^{(t,x)} = 1 | X_1^{(t,x)}, \ldots, X_{i-1}^{(t,x)}] \leq  \P[{X}_i^{(t,x)} = 1] \leq  \P[\hat{X}_i^{(t,x)} = 1],$$
in  Claims~\ref{claim:history_wont_help} and  \ref{claim:upperbound_prob_exp}, respectively. This would imply  $\sum_{X^{(t,x)}_i\in{\mathcal{X}}^{(t,x)}}{X^{(t,x)}_i} \preceq \sum_{\hat{X}^{(t,x)}_i\in\hat{\mathcal{X}}^{(t,x)}}\hat{X}^{(t,x)}_i$ by  \Cref{lem:doerr} above,
and so $\overline{\deg}_{A^t}(x) \preceq \widehat{\deg}_{A^t}(x)$, completing the proof \Cref{lem:bound deg bar}.

% \begin{definition}
% 	% is this needed?
% 	For a fixed job $u$ and a fixed timestep $t$, we let $p(t,u) = 1 /\deg_{G^t}(u)$ be the probability that
% 	$\Sampling(u)$ chooses $e$ at time $t$. We also define $p(t,e) = p(t,\job(e))$.
% \end{definition}

%We now inspect the probability that $X_i^{(t,x)}$ is being $1$ more carefully.
%By definition and by our algorithm,
%$\P[X_i= 1] = \frac{1} { \deg_{ G^{t'} } (u) } \leq \frac{1} { \deg_{G^{t}} (u ) }$
%where $u = \job(X_i)$ and $t > t' = t(X_i)$.
%% p(t(X_i),e) \leq p(t',e)$ if $t' > t(X_i)$.
%Below, we show that this inequality still holds true even if we consider $(t,x)$-relevant experiments.
% The $(t,x)$-relevant experiment $X_i^{(t,x)}$ can taken place at timestep$t'=t(X_i^{(t,x)}) \leqt$, hence it is possible that $\deg_{G^{t'}}(x) >   \deg_{G^t}(x).$ However, as $\deg_G(x)$ is decreasing, $p(t',e) \leq p(t,e)$, so we should be able to use $p(t,e)$ as an upperbound to $\P[X_i^{(t,x)}=1].$ We prove this intuition more formally below.

% \begin{definition}
%     % is this needed?
%     For a fixed edge $e$ and a fixed timestep $t$, we let $p(t,e) = 1 /\deg_{G^t}(job(e))$. be the upperbound of the probability we sampling an edge $e$ at time $t$.
% \end{definition}

\begin{claim}
	\label{claim:upperbound_prob_exp}
	For any $(t,x)$-relevant experiment $X^{(t,x)}_i$,
	$\P[X_i^{(t,x)} = 1] \leq \P[\hat{X}_i^{(t,x)} = 1]$.
	% \thatchaphol{Just claim $\P[X_i^{(t,x)} = 1] \le \P[\hat X_i^{(t,x)} = 1] $ and put $p(t,e)$ in the proof}
\end{claim}
\begin{proof}
Let $e = e(X_i^{(t,x)})$ and $t' = t(X_i^{(t,x)})$ be the time that the experiment $X_i^{(t,x)}$ is taken. Note that $t' \le t$ as $X_i^{(t,x)}$ is $(t,x)$-relevant.
The claim follows because
$$
\P[X_i^{(t,x)} = 1] = \frac{1}{ \deg_{G^{t'}}(\job(e)) }
\leq \frac{1}{ \deg_{G^{t}}(\job(e)) }
= \P[\hat{X}_i^{(t,x)} = 1].
$$
The first equality is because $X_i^{(t,x)}$ succeeds iff $\Sampling(\job(e))$ chooses $e$ at time $t'$, which happens with probability $1/ \deg_{G^{t'}}(\job(e))$. (Note that knowing that $X_i^{(t,x)}$ is $(t,x)$-relevant does not change the probability that the experiment $X_i^{(t,x)}$ succeeds because we can determine if $X_i$ is $(t,x)$-relevant depends on information in the past, including $X_1, X_2, \ldots, X_{i-1}$ and the adversary's actions.)
The inequality is because $t' \le t$ and $G$ undergoes deletions only. The last equality is by definition of $\hat{X}_i^{(t,x)}$.
%
%
%Knowing that $X_i$ is $(t,x)$-relevant does not change the probability that the experiment $X_i$ succeeds, i.e., $\P[X_i | X_i \text{ is } (t,x)\text{-relevant}] = \P[X_i]$.
%
%	Notice that, since $X_i^{(t,x)}$ is $(t,x)$-relevant, $t' = t(X_i^{(t,x)}) \leq t$.
%	Since $\deg_{G^{t'}}(\job(e)) \geq \deg_{G^t}(\job(e))$, we have
%	$$
%		\P[X_i^{(t,x)} = 1] = \frac{1}{ \deg_{G^{t'}}(\job(e)) }
%		\leq \frac{1}{ \deg_{G^{t}}(\job(e)) }
%		= \P[\hat{X}_i^{(t,x)} = 1].
%	$$
%	% \begin{align*}
%	% 	\P[X_i^{(t,x)} = 1]
%	% 	& = \frac{1}{ \deg_{G^{t'}}(\job(e)) } \\
%	% 	& \leq \frac{1}{ \deg_{G^{t}}(\job(e)) } \\
%	% 	& = \P[\hat{X}_i^{(t,x)} = 1].
%	% \end{align*}
\end{proof}

% The assumption below, which is true for the instance we generated from our $3$-spanner problem, guarantees that this subtle issue cannot occur.

\begin{claim}
	\label{claim:history_wont_help}
	$\P[X_i^{(t,x)} = 1 | X_1^{(t,x)}, \ldots, X_{i-1}^{(t,x)}] \leq \P[X_i^{(t,x)} = 1].$
\end{claim}
\begin{proof}
	By \Cref{assump:disjoint}, this is true simply because knowing the results of the past experiments and other experiments taken at the same timestep as $X_i^{(t,x)}$  cannot increase the
	probability that $X_i^{(t,x)}$ being $1$. Without the assumption, for some $i$ and $j<i$,
	it is possible that $\P[X_i^{(t,x)} = 1 | X_{j }^{(t,x)} = 0] > \P[X_i^{(t,x)} = 1]$
	if $t(X_i^{(t,x)}) = t(X_j^{(t,x)})$ and $\job(X_i^{(t,x)}) = \job(X_j^{(t,x)})$, i.e.,
	$X_i^{(t,x)}$ and $X_j^{(t,x)}$ are being sampled with the same $\Sampling(\job(X_i^{(t,x)})$ call.
\end{proof}

\subsection{Key Step 3}
\label{sub:key 3}
In this section, we prove \Cref{lem:bound deg hat}.
To simplify our proofs, we say that  \emph{time $t'$ is $(t,e)$-relevant} if there is a $(t,e)$-relevant experiment $X$ created at time $t(X) = t'$.
Since, for each time step $t'$, the algorithm can only create one  $(t,e)$-relevant experiment, we have the following observation:
\begin{observation}\label{obs:rel time}
	The number of $(t,e)$-relevant experiments $R(t,e)$ is exactly the number of $(t,e)$-relevant time steps.
\end{observation}

Now, we state the following crucial lemma. It says that, there are at most $\log (t)$ experiments that are $(t,e)$-relevant.
\begin{lemma}
	\label{claim:log_relevant}
	$Rel(t,e) \leq \log (t)$ for every $t$ and $e\in R^t$.
\end{lemma}
\begin{proof}
	% In this proof, we say that time $t'$ is $(t,e)$-relevant if there is a $(t,e)$-relevant experiment $X$ where $t(X) = t'$.
	By \Cref{obs:rel time}, we will bound the total number of $(t,e)$-relevant time steps.
	Suppose $t'$ is $(t,e)$-relevant. It suffices to show that if there is another time $t''>t'$ which is $(t,e)$-relevant, then $t'' \geq (t' + t)/2$. This means that
	each consecutive $(t,e)$-relevant time steps decrease the gap to the fixed time step $t$ by at least half. So this can happen at most $\log(t)$ times.

	To prove the claim, observe that $t'' \notin \Schedule^{t'}(\job(e))$ as $t'$ is $(t,e)$-relevant.
	 Hence, the adversary must touch $\job(e)$ at some timestep $s\geq t'$. When that happens, we add $s+1, s+2, \ldots$ into $\Schedule^{s}(\job(e))$. Let $s' = s +2^{\log {(t-s)}-1}$. It is clear that
	$$s' \geq s + 2^{\lceil \log (t-s) \rceil -1} \geq (s+t)/2 \geq (t' + t )/2.$$
	Because any timestep in $(t',s']$ cannot be $(t,e)$-relevant, $t''$ must be greater than $s'$.
	Hence, $t'' \geq (t' + t)/2$ as claimed.
\end{proof}

The above implies that the expected value of $\widehat{\deg}_{A^t}(x)$ is not too far from the target load of $x$ at time $t$.

\begin{lemma}\label{lem:expect deg hat}
	$\E [\widehat{\deg}_{A^t}(x)  ]  \leq \log{(t)}\target^t(x)$.
\end{lemma}
\begin{proof}
	We have the following
%
%	We first show that $\E[  \widehat{\deg}_{A^t}(x)  ]  \leq \log{(t)}\target^t(x)$. By \Cref{claim:log_relevant}, each edge $e$ has at most $O(\log{(t)})$ $(t,e)$-relevant experiments.
%	% Notice that if $e,e'$ be such that $\job(e) = \job(e')$,
%	% then the set of timesteps that $(t,e)$-relevant and $(t,e')$-relevant are identical.
%	Hence, for any job $u$ adjacent to machine $x$, there are at most $O(\log{(t)})$ different $t'$
%	% such that $(t',u) \in Rel(\mathcal X^{(t,F)})$.
%	such that $t'$ is $(t,x)$-relevant.
%	Since $\E[\hat{X_i}^{(t,x}] =  p(t,e(\hat{X_i}^{(t,x} ))$,
%	it follows that
	$$\E[\widehat{\deg}_{A^t}(x) ] = \E[  \sum_{ \hat X \in \hat {\mathcal X} ^{(t,x)}} \hat X ] \leq \log{(t)} \sum_{e \in R^t(x)} 1/\deg_{G^t}(\job(e)) = \log(t)\cdot \target^t(x),$$
where the first and last equalities are by definitions of $\widehat{\deg}_{A^t}(x)$  and $\target^t(x)$, respectively. It remains to prove the inequality.

Observe that $\sum_{\hat{X}\in\hat{\mathcal{X}}^{(t,x)}\mid e(X_{i}^{(t,x)})=e}\E[\hat{X}]=Rel(t,e)/\deg_{G^{t}}(\job(e))$.
This is because the number of terms in the sum is exactly the number
of $(t,e)$-relevant experiments $Rel(t,e)$, and we precisely define
each $\hat{X}_{i}^{(t,x)}\in{\cal \hat{X}}^{(t,x)}$ where $e=e(X_{i}^{(t,x)})$
so that $\E[\hat{X}_{i}^{(t,x)}]=\P[\hat{X}_{i}^{(t,x)}=1]=1/\deg_{G^{t}}(\job(e))$.
Therefore, by \Cref{claim:log_relevant}, we have
$$ \E[\sum_{\hat{X}\in\hat{\mathcal{X}}^{(t,x)}}\hat{X}]=\sum_{e\in R^{t}(x)}\sum_{\hat{X}\in\hat{\mathcal{X}}^{(t,x)}\mid e(X_{i}^{(t,x)})=e}\E[\hat{X}]\le \log(t)\cdot\sum_{e\in R^{t}(x)}1/\deg_{G^{t}}(\job(e)).$$
\end{proof}

The last step is to show that the expectation bound from \Cref{lem:expect deg hat} is concentrated. However, since $\target(x)$ for the machine $x$ can be very small ($\ll1$),
it is not enough to use the standard multiplicative Chernoff's bound.
Instead, we will apply the version with both additive error and multiplicative error stated below.
% However, it is not natural since many random variables in $\hat{\mathcal Y}^{(t,F)}$ concern the same job at different times.

% \begin{definition}
%     Given $\mathcal Y^{(t,F)}$, for each $j \in J$, we let $W_j^{(t,F)}$ be a boolean random variable that is $1$ if there exists $i$ such that $\job(Y_i^{(t,F)})=j$ and $Y_i^{(t,F)}$ is $1$.
%
%     Let $\mathcal W^{(t,F)} = \sum_{e \in F} W_e^{(t,F)}$.
% \end{definition}
%
% \begin{observation}
%     $\P[W_j^{(t,F)} = 1] = \E[W_j^{(t,F)}] \leq \log{(t)}\ p(t,j)$.
% \end{observation}
%
% \begin{proof}
%
% \end{proof}

% For each edge $e \in F$, we define $W_e^{(t,F)}$ to be a random variable that is true if any experiment concerns $e$ in $\mathcal Y^{(t,F)}$ is true. Note that $\E[W_e^{(t,F)}] \leq \log{(t)} p(t,e)$.
% Let $\mathcal W^{(t,F)} = \sum_{e \in F} W_e^{(t,F)}$.

\begin{lemma}[Additive-multiplicative Chernoff's bound~\cite{soda/BadanidiyuruV14}]
	Let $X_1, \ldots, X_n$ be independent binary random variables.
	Let $S = \sum_{i =1}^n X_i$. Then for all $\delta \in[0,1]$ and $\alpha > 0$,
	$$\P[S \geq (1+\delta) \E[S] + \alpha ]\leq \exp( - \frac{\delta \alpha}{3}).$$
\end{lemma}

% \begin{claim}
% 	\label{claim:machine_load_concentration}
% 	Let $x$ be a machine.
% 	We have that, for any $k> 0$
% 	$$\P\left[  \widehat{\deg}_{A^t}(x) \geq 2 \log{(t)} \target(x) + k\log{|M|}\right] \leq |M|^{-k/3}
% 	.$$
% \end{claim}

\paragraph*{Proof of \Cref{lem:bound deg hat}.}
By
plugging $\delta = 1$ and $\alpha = 30\log{|M|}$ into the above bound, we have
$$\P\left[\widehat{\deg}_{A^t}(x) \geq 2 \E[\widehat{\deg}_{A^t}(x)] + 30\log{|M|}\right] \leq
\exp( - 30\log{|M|} / 3) = |M|^{-10}.$$
This completes the proof of \Cref{lem:bound deg hat} because $\E [\widehat{\deg}_{A^t}(x)  ]  \leq \log{(t)}\target^t(x)$ by \Cref{lem:expect deg hat}.

%!TEX root=main_writeup.tex
\section{Lifting the Machine-Disjoint Assumption}
\label{sec:lift_assumption}

In the proof of \Cref{lemma:overhead} in \Cref{sec:reduce to load},
we assume \Cref{assump:disjoint} which says that, for a fixed job $u$, routines for $u$
are machine-disjoint. In this section, we show how to remove it.

\paragraph*{Reiterating the pain point.}
Each experiment $X_i$ is defined \textbf{with respect to} a sampling on each edge.
In any sequence $\mathcal X^{(t,x)}$, there could be two $(t,x)$-relevant experiments
$X_i^{(t,x)}$ and $X_{j}^{(t,x)}$ where $i<j$ such that
$t(X_i^{(t,x)}) = t(X_j^{(t,x)})$ and $\job(X_i^{(t,x)}) = \job (X_j^{(t,x)})$.
By our sampling process, these two variables are negatively correlated
as knowing $X_i^{(t,x)} = 1$ would immediately imply that $X_j^{(t,x)} = 0$ and vice versa.

To see why the condition in \Cref{lem:doerr} may not be true in this case,
let consider the case where we have two random variables $X$ and $Y$ where $\P[X=1] = p$ be the probability that $X$ being $1$.
If these two variables are correlated in such a way that $X+Y = 1$, then we cannot show that  $\hat X , \hat Y$ such that $X + Y \preceq \hat X + \hat Y$ using \Cref{lem:doerr} unless $\P[\hat X] = \P[\hat Y] = 1$.
We want to make sure that variables in the sequence we want to analyze do not have
such correlation.

\paragraph*{Proving strategy.} 
%It is not hard to see that, the assumption is crucial only to prove \Cref{lem:bound deg bar}. 
Here, we will define two other sequences $\mathcal Y^{(t,x)}$ and $\hat{\mathcal Y}^{(t,x)}$. These two sequences are defined analogous to $\mathcal X^{(t,x)}$ and
$\hat{\mathcal X}^{(t,x)}$.
We then define $\widetilde{\deg}_{A^t}(x) = \sum_{\hat Y \in \hat{\mathcal Y}^{(t,x)}} \hat Y$.
We will replace $\widehat{\deg}_{A^t}(x)$ in Key Steps 2 and 3 by $\widetilde{\deg}_{A^t}(x)$.
More precisely, instead of \Cref{lem:bound deg bar}, we will prove $\overline{\deg}_{A^t}(x) \preceq \widetilde{\deg}_{A^t}(x)$, and 
we will prove that, w.h.p.,
$\widetilde{\deg}_{A^t}(x) \leq 2\target^t(x) + O( \log |M|) $ instead of \Cref{lem:bound deg hat}. 
% With these two properties combined, it is easy to see that \Cref{lem:bound deg bar} holds.

\paragraph*{Redefining relevant experiments.} For a fixed $\mathcal X^{(t,x)}$,
we let $\pi^{(t,x)}(t',u) = \{ X_i^{(t,x)} : t(X_i^{(t,x)}) = t', \job(X_i^{(t,x)}) = u  \}$
be the set of $(t,x)$-relevant experiments containing $u$ taken at time $t'$.
We then let $Y^{(t,x)}(t',u)$ be a boolean random variable that is one
if any of experiments in $\pi^{(t,x)}(t',u)$ succeed. In other words,
$$ Y^{(t,x)} (t',u) = \sum_{X \in \pi^{(t,x)}(t',u)} X.$$
Note that $Y^{(t,x)}(t',u)$ is a boolean random variable because
only at most one variable in $\pi^{(t,x)}(t',u)$ can be true as
they are being sampled with the same $\Sampling(u)$ call.
We say that $Y^{(t,x)}(t',u)$ is \textbf{well-defined} if $|\pi^{(t,x)}(t',u)|>0.$
Note that $$\P[Y^{(t,x)}(t',u) = 1] = \frac{|\pi^{(t,x)}(t',u)|} { \deg_{G^{t'}}(u)}.$$

Now let us define more notations to help count the number of relevant experiments in form of $Y^{(t,x)}(.,.)$.
\begin{definition}
    We say that a timestep $t'$ is $(t,x,u)$-relevant if
    $Y^{(t,x)}(t',u)$ is well-defined.
    % there exists a $(t,x)$-relevant experiment $X$ where $t(x) = t'$ and $\job(X) = u$.
    Let $Rel(t,x,u)$ be the number of $(t,x,u)$-relevant timesteps, this is exactly
    the number of possible $t'$ such that $Y^{(t,x)}(t',u)$ is well-defined.
    We also let $\widetilde{Rel}(t,x) = \sum_u Rel(t,x,u)$ be the number of possible pairs of
    $t',u$ such that $Y^{(t,x)}(t',u)$ is well-defined.
\end{definition}
We then define a sequence $\mathcal Y^{(t,x)} = Y_1^{(t,x)}, Y_2^{(t,x)}, \ldots, Y_{\widetilde{Rel}(t,x)}^{(t,x)}$ to be the sequence of well-defined $Y^{(t,x)}(.,.)$ order by the time of these experiments. Notice that $\overline{\deg}_{A^t}(x) = \sum_{X \in \mathcal X^{(t,x)}} X = \sum_{Y \in \mathcal Y^{(t,x)}} Y.$

Now we discuss the replacement of $\hat{\mathcal X}^{(t,x)}$.
For each $Y_i^{(t,x)}$, we let $\hat{Y}_i^{(t,x)}$ be an independent random variable that is true with probability $ \frac{ |\pi^{(t,x)}(t',u)| }{\deg_{G^t}(u)}$.
Then the sequence $\hat{\mathcal Y}^{(t,x)}$ is  $\hat{Y}_1^{(t,x)}, \hat{Y}_2^{(t,x)}, \ldots, \hat{Y}_{\widetilde{Rel}(t,x)}^{(t,x)}$.
As mentioned above, $\widetilde{\deg}_{A^t}(x) = \sum_{\hat Y \in \hat{\mathcal Y}^{(t,x)}} \hat Y$.
% Note that $$\E[\widehat{\deg}_{A^t}(x)] = \E[\widetilde{\deg}_{A^t}(x) ] $$
% and $$\P[\widehat{\deg}_{A^t}(x) > \lambda ] \geq \P[\widetilde{\deg}_{A^t}(x)> \lambda].$$

We state the two main claims below.

\begin{claim}[Key Step $2^\prime$]
    $\overline{\deg}_{A^t}(x) \preceq \widetilde{\deg}_{A^t}(x)$.
\end{claim}
\begin{proof}
    % This follows essentially the proofs
As $\overline{\deg}_{A^t}(x) = \sum_{Y \in Y^{(t,x)}} Y$,
by \Cref{lem:doerr}, we need to show that
$\P[Y_i^{(t,x)}  = 1 | Y_1^{(t,x)}, \ldots, Y_{i-1}^{(t,x)}] \leq \P[\hat Y_i^{(t,x)}].$
By replacing $\mathcal X^{(t,x)}$ and $\hat{\mathcal X}^{(t,x)}$ with
$\mathcal Y^{(t,x)}$ and $\hat{\mathcal Y}^{(t,x)}$,
the proof here can be done with the same arguments we used in \Cref{sub:key 2}. Note that the proof works because random variables from $\mathcal Y^{(t,x)}$ are \emph{not} negatively correlated.
%Hence, we omit the rest of the proof.
\end{proof}

\begin{claim}[Key Step $3^\prime$]
	$ \widetilde{\deg}_{A^t}(x) \le 2 \log{(t)} \target^t(x) + O(\log{|M|})$ with probability $1-1/{|M|}^{10}$.
    % $\widetilde{\deg}_{A^t}(x) \preceq \widehat{\deg}_{A^t}(x).$

\end{claim}
\begin{proof}
    We can follow exactly the proof in \Cref{sub:key 3}. The only crucial part
    is to show that $\E[\widetilde{\deg}_{A^t}(x)] = \E[\widehat{\deg}_{A^t}(x)]$.
    The other parts can be argued exactly the same way.

    % It suffices to show that
    % $$\P[\widetilde{\deg}_{A^t}(x) > \lambda] \leq \P[\widehat{\deg}_{A^t}(x) > \lambda].$$
    %
    This is true by the way we define the process. For each $\hat Y_i^{(t,x)}$ that is being $1$
    with probability $c/\deg_{G^t}(\job(Y_i^{(t,x)})$ for some integer $c$, we have
    $c$ different boolean variables $\hat X_1,\hat X_2,\ldots,\hat X_c$ in $\hat{\mathcal X}^{(t,x)}$ that is being $1$
    with probability $1/\deg_{G^t}(\job(Y_i^{(t,x)})).$ Hence, the two expectations are identical
    by linearity of expectation.
\end{proof}

% For a fixed time $t$, a fixed machine $x$, and a fixed job $u$,
% let $X(t,x,u) = \{X_i : t(X_i)= t, m \in e(X_i), \jpb(X_i) = u\}$, be the
% set of experiments of edges in

% Let $\mathcal X^{(t,x,u)} = X_1^{(t,x,u)}, X_2^{(t,x,u)}, \ldot,

%!TEX root=main_writeup.tex
\section{A Fully-dynamic-to-decremental Reduction for Spanners}\label{sec:reduction}

In this section, we give a reduction from a fully-dynamic spanner to a decremental spanner.
This reduction is due to~\cite{BaswanaKS12} and we provide it here for the
completeness of our paper.

Let $E_1 \ldots E_j$ be a partition of $E$, then the observation below states that
the union of spanners of $E_1 \ldots E_j$ is a spanner of $E$.

\begin{observation}[Observation~5.2 in~\cite{BaswanaKS12}]
    \label{obs:union_spanners}
    For a given graph $G = (V,E)$, let $E_1 \ldots E_j$ be a partition of the set of edges $E$,
    and let $\mathcal{E}_1 \ldots \mathcal{E}_j$ be respectively the $t$-spanner of subgraphs
    $G_1 = (V,E_1), \ldots, G_j = (V,E_j)$. Then $\bigcup_i \mathcal{E}_i$ is a
    $t$-spanner of the original graph $G=(V,E)$.
\end{observation}

With this observation, the idea behind the reduction is to split into $O(\log n)$ subgraphs in
such a way that every subgraph, excepts one subgraph, is a decremental instance.

Formally, let $\ell_0$ be the greatest integer such that $2^{\ell_0} \leq n^{1+1/k}$. We do the following.
\begin{enumerate}
 \item We partition $E$ into $E_0 \ldots E_j, j = \lceil \log_2 n^{1-1/k} \rceil$ such that $|E_i| \leq 2^{\ell_0 + i}$. Each edge will belong to only one set $E_i$ and we keep track of this information.
 \item For each $E_i$, $i>0$, we maintain $H_i = (V,\mathcal E_i)$, which is a
 $(2k-1)$-spanner of $(V,E_i)$.
 \item We maintains a binary counter $\mathbf{C}$ which counts from $0$ to $\frac{n(n-1)}{2}$.
 This will be used to decide when to rebuild $E_i$.
\end{enumerate}

In the beginning, we set $E_j = E$ and $E_i = \emptyset$ for all $i < j$. The counter $\mathbf C$ is set to $0$.
Any edge deletion of $e \in E_i$ for any $i$ is handled as in the decremental case.
When an edge $e$ is inserted, we increment the counter $\mathbf{C}$ by one. Let $g$ be the highest bit of $\mathbf C$ that gets flipped.
If $g \leq \ell_0$, then we put $e$ in $E_0$ and $\mathcal E_0$. Otherwise, we insert $e$ into $E_h$, where $h = g - \ell_0$,
move all edges from $E_i$, $i<h$ to $E_h$. At this moment $E_i = \emptyset$ for all $i <h$. We then rebuilt the spanner $\mathcal E_h$.

From Observation~\ref{obs:union_spanners}, $\bigcup_i \mathcal E_i$ is a $(2k-1)$-spanner of $G$.
% we get the following result.

% \begin{theorem}
%     There exists an algorithm that maintains $(2k-1)$-spanner $H = (V,E_H)$
%     of $G= (V,E)$ of size $O(n^{1+1/k} \log n)$
%     in a fully-dynamic graph with $O(m \log n)$ total recourse.
% \end{theorem}
\begin{lemma}[Restate Lemma~\ref{lemma:fully_dyn_reduction}]
	\label{lemma:fully_dyn_reduction_restate}
	Suppose that for a graph $G$ with $n$ vertices and $m$ initial edges undergoing only edge deletions, there is an algorithm that maintains a $(2k-1)$-spanner $H$ of size $O(S(n))$  with $O(F(m))$ total recourse where $F(m) = \Omega(m)$,
	then there exists an algorithm that maintains a $(2k-1)$-spanner $H'$ of size $O(S(n) \log n)$
	in a fully dynamic graph with $O(F((U) \log n))$ total recourse. Here $U$ is the number of updates made throughout the algorithm, starting from an empty graph.
\end{lemma}
\begin{proof}[Proof of Lemmma~\ref{lemma:fully_dyn_reduction}]
    We use the reduction above to partition the graph into $E_1, \ldots, E_j$.
    We then use the decremental algorithm to maintain each $\mathcal E_i$ for all $i > 0$.

    We first show that the size of $|H'| = |\bigcup_i \mathcal E_i| = O(\log n S(n))$.
    As we have $O(\log n)$ subgraphs, $|\bigcup_i \mathcal E_i| = \sum_i S(|E_i|) \leq O( S(n)\log n)$.
    % Combining this with Theorem~\ref{thm:greedy_spanner},

    Now we show the total recourse.
    Let $\mathcal G_1, \mathcal G_2, \ldots, \mathcal G_k$ be all the graph we rebuilt throughout all the timesteps.
    Then the total recourse is bounded by $\sum_i F(|\mathcal G_i|) \leq F(\sum_i |\mathcal G_i|)$.
    Notice that the level of any edge $e$ can only go up, so $e$ can contribute the recourse to only $\log n$ different graphs. Hence, this inequality becomes
    $$F(\sum_i |\mathcal G_i|) \leq F( (U) \log n ).$$
    % Each time
    % and since
    % the algorithm in Theorem~\ref{thm:greedy_spanner} never removes
    % any spanner edge, each edge can be inserted only at most $j = O(\log n)$ times.
    % Hence, the total recourse is $O(m \log n)$.
\end{proof}

%!TEX root=ICALP_main.tex
\section{Missing Proofs from \Cref{sec:3spanner}}
\label{sec:app:missing:3spanner}

\begin{claim}[Restate \Cref{cl:static:stretch-size}]
    The subgraph $H = (V, E_1 \cup E_2 \cup E_3)$ is a $3$-spanner of $G$ consisting of at most $O(n\sqrt{n})$ edges.
\end{claim}
\begin{proof} We need to show that (1) $H$ is a $3$-spanner and (2) $E_H = E_1 \cup E_2 \cup E_3$ has size at most $O(n \sqrt{n})$.
GG
        \paragraph*{Stretch.} Consider an edge $e = (u,v)$ where $u \in V_i, v \in V_j$.
    We show that $H$ has a path of length at most 3 between $u$ and $v$.
    The easy case is when $(u,v) \in E_H$.
    This gives us a path of one edge.
    It happens when $u = c_i(v)$ or $v = c_j(u)$, or $i = j$.
    Suppose $(u,v ) \notin E_H$.
    Consider $v' = c_j(u)$. Since $u$ is a common neighbor between $v$ and $v'$,
    $P(v,v')$ is not empty. As $e \notin E_H$, $u \neq w_{vv'}$.
    As, the path $u, v', w_{vv'}, v$ has length exactly $3$,
    the stretch part is concluded.

    \paragraph*{Size.} Each vertex $u$ has upto $\sqrt{n}$ partners.
    Since we have $n$ vertices, $|E_1| = O(n \sqrt{n})$.
    For $E_2$, the graph induced on $V_i$ has at most $O(n)$ edges.
    Since we have $\sqrt{n}$ buckets, $|E_2| = O(n \sqrt{n})$.
    For $E_3$, $|E_3|$ is bounded by the number of witnesses we need.
    Since we have $O{\sqrt{n}}$ buckets, and we have $O(n)$ pairs of
    vertices within the same bucket, for all buckets, $|E_3|$ must
    be bounded by $O(n \sqrt{n})$.
    We conclude the proof by saying that $|E_H| = O(|E_1| + |E_2| + |E_3|)
    = O(n\sqrt{n})$

\end{proof}

\subsection{Making the update time worst case}
\label{sub:app:spread:work}

It is evident that the update time in Lemma~\ref{lm:worstcase} is in worst-case. The whole algorithm is amortized only because we are {\em replacing} $E_3$ from scratch at the start of each phase, which takes $\tilde{O}(n^2)$ time, and this time needs to be amortized over the length of the phase.
% In particular, within a phase, we actually spend $\tilde{O}(\sqrt{n})$ time in the worst case per update (see Lemma~\ref{lm:worstcase}).
Using very standard techniques from the existing literature on  dynamic algorithms (see e.g.~\cite{thorup2005worst,baswana2016dynamic,NanongkaiSW17,kiss2021deterministic}), we can easily convert this into an overall worst case update time guarantee.
The idea is to {\em spread out} the $\tilde{O}(n^2)$ cost of rebuilding at the start of a given phase through a sufficiently large chunk of updates in the phase preceding it.
% We leave the remaining details for the full version.

For completeness, we will describe the idea in more detail here.
Let $G_i$ be the graph after $i$ updates.
We will use one instance of our algorithm to handle
$L = \tilde{\Theta}(n\sqrt{n})$, that is,
the $i$-th instance will be used to handle the time steps $[(i-1)L, iL)$.
For our idea to work, any copy must be able to handle $2L$ updates (which is fine in our case).
At the beginning, we initiate the first copy $D_1$ with time $\Otil(n^2)$.
We want to initiate $D_2$, as well as feed $D_2$ with $L$ updates, so that $D_2$ is ready to use
at the timestep $L$.
This can be done in the following manner.
\begin{itemize}
    \item During time steps $[0,L/3)$, we initiate $D_2$ with the graph $G_0$,
    \item During time steps $[L/3, 2L/3)$, we carefully feed the surviving\footnote{Some edges in the spanner of $D_2$ might be deleted between the time step $[0,2L/3]$, we must not add deleted edges in our actual output.}
    output from $D_2$ into our actual output,
    \item During time steps $[2L/3,L)$, we update $D_2$ with updates from time steps $[0,L)$, $3$ updates at a time.
\end{itemize}

Hence, at time step $L$, we can switch from $D_1$ to $D_2$, disregard $D_1$, and start initiating $D_3$.
More generally, suppose at time step $iL$, after we have initiated $D_i$, which we will use for time step $[iL, (i+1)L)$.
In the following $L$ time steps, to disregard $D_{i-1}$ and initiate $D_{i+1}$, we do the following.
\begin{itemize}
    \item During time steps $[iL,iL + L/3)$, we initiate $D_{i+1}$ with the graph $G_{iL}$ and slowly disregard the output from $D_{i-1}$,
    \item During time steps $[iL + L/3,iL + 2L/3)$, we carefully feed the surviving output from $D_{i+1}$ into our actual output,
    \item During time steps $[iL + 2L/3,(i+1)L)$, we update $D_{i+1}$ with updates from time steps $[iL,(i+1)L)$, $3$ updates at a time.
\end{itemize}

Then, at time step $(i+1)L$, we completely disregard $D_{i-1}$, and completely initiate $D_{i+1}$, hence maintaining all the desired properties.
Our actual output at time step $t$, is consisting of $E_{D_i}$ and $E_{D_{i+1}}$, which are output of $D_i$ and $D_{i+1}$, respectively.
Notice that, $(V,E_{D_i})$ is a $3$-spanner of $G_t$, hence, $(V, E_{D_i} \cup E_{D_{i+1}})$ is also a $3$-spanner of $G_t$.
This is because a spanner with extra edges is still a spanner. Hence, as long as $|E_{D_i} \cup E_{D_{i+1}}|$ is not too large, we can keep both of them in our output at the same time.
We conclude by saying that, this idea of maintaining multiple copies, along with our \Cref{lm:worstcase:updatetime}, implies \Cref{thm:main det worst case}.

\textbf{Note.} As a corollary, we note that the idea above does not have anything to do with Spanner. Rather, it is applicable to any dynamic problem, as long as, it can handle batch updates in worst-case time, and that the bottleneck of the algorithm is on the initialization time.

\end{document}